\documentclass[journal]{IEEEtran}

\usepackage{amsfonts}
\usepackage{setspace}

\usepackage{xcolor}
\usepackage{amsfonts}
\usepackage{setspace}
\usepackage{amsmath}
\usepackage{amssymb}
\usepackage{amsthm}
\usepackage[pdftex]{graphicx}
\usepackage{epsfig}
\usepackage[justification=centering]{caption}

\usepackage[ps2pdf,
bookmarks=false,
bookmarksnumbered=false, 
bookmarksopen=false, 
colorlinks=false]{}

\newcommand{\E}{\mathbb{E}}

\newcommand{\figsize}{1.05}

\newcommand{\Nreq}{N_{\text{req}}}
\newcommand{\tth}{\text{th}}

\newcommand{\BS}{{\mathbf{BS}}}

\newtheorem{Lem}{Theorem}

\begin{document}

%

\title{Learning-Based Delay-Aware Caching in Wireless D2D Caching Networks}

\author{{Yi Li, Chen Zhong, M. Cenk Gursoy and Senem Velipasalar}\\
\thanks{Y. Li is with Intelligent Fusion Technology, Germantown, MD (e-mail: yli33@syr.edu).
C. Zhong, M. C. Gursoy, and S. Velipasalar are with the Department of Electrical Engineering and Computer Science,
Syracuse University, Syracuse, NY 13244 (e-mail: czhong03@syr.edu, mcgursoy@syr.edu, svelipas@syr.edu).}}


\maketitle

\thispagestyle{empty}

\begin{abstract}
Recently, wireless caching techniques have been studied to satisfy lower delay requirements and offload traffic from peak periods. By storing parts of the popular files at the mobile users, users can locate some of their requested files in their own caches or the caches at their neighbors. In the latter case, when a user receives files from its neighbors, device-to-device (D2D) communication is performed. D2D communication underlaid with cellular networks is also a new paradigm for the upcoming wireless systems. By allowing a pair of adjacent D2D users to communicate directly, D2D communication can achieve higher throughput, better energy efficiency and lower traffic delay. In this work, we propose an efficient learning-based caching algorithm operating together with a non-parametric estimator to minimize the average transmission delay in D2D-enabled cellular networks. It is assumed that the system does not have any prior information regarding the popularity of the files, and the non-parametric estimator is aimed at learning the intensity function of the file requests. An algorithm is devised to determine the best $<$file,user$>$ pairs that provide the best delay improvement in each loop to form a caching policy with very low transmission delay and high throughput. This algorithm is also extended to address a more general scenario, in which the distributions of fading coefficients and values of system parameters potentially change over time. Via numerical results, the superiority of the proposed algorithm is verified by comparing it with a naive algorithm, in which all users simply cache their favorite files, and by comparing with a probabilistic algorithm, in which the users cache a file with a probability that is proportional to its popularity.
\end{abstract}



\IEEEpeerreviewmaketitle


\section{Introduction}

Mobile data traffic has grown more than 18-fold over the past five years \cite{Cisco}. In addition, as demonstrated by the survey results provided in \cite{Cisco}, smartphones represented only $45\%$ of the total mobile devices and connections in 2016, but generated $81\%$ of the total mobile traffic. In the same year, smartphone usage also grew by $38\%$. Moreover, these trends are also expected to continue in the near future. As predicted in \cite{Cisco}, by 2021, annual global mobile data traffic will exceed half a zettabyte,  representing a growth of $7$ times of that in 2016. This growth is primarily fueled by the mobile video traffic. Indeed, video, which is cacheable, already accounts for $60\%$ of the mobile data traffic and is expected to be $78\%$ of the mobile data traffic by 2021. At the same time, the storage capacity of mobile devices is growing, leading to the availability of more cache space. These trends and predictions motivate the deployment of wireless device-to-device (D2D) caching networks in which popular contents can be cached beforehand and D2D transmissions can be enabled to exchange cached files with the goal to achieve lower delays in communication while offloading traffic from peak periods and congested links.

Recently, many studies have been conducted to analyze caching strategies in wireless networks in order to satisfy the throughput, energy efficiency and latency requirements in next-generation 5G wireless systems. By storing parts of the popular files at the base station and users' devices, network traffic load can be managed/balanced effectively, and traffic delay can be greatly reduced. A brief overview of wireless caching was provided in \cite{WL_caching_overview}, which introduced the key notions, challenges, and research topics in this area. In order to improve the performance effectively, the system needs to learn and track the popularity of those cacheable contents, and predict the popularity variations, helping to guarantee that the most popular contents are cached and the outdated contents are removed. In \cite{LivingOnTheEdge}, popularity matrix estimation algorithms were studied for wireless networks with proactive caching.

Multiple caching strategies have been investigated in the literature, which improve the performance in different ways. When contents are cached at the base stations, the energy consumption, traffic load and delay of the backhaul can be reduced \cite{ClusterContentCaching}, and the base stations in different cells can cooperate to improve the spectral efficiency gain \cite{PhyCaching5G}. When contents are cached at the users' devices, the base station can combine different files together and multicast to multiple users, and the users can decode their desired files using their cached files. A content distribution algorithm for this approach was given in \cite{MulticastingCodingRandomDemands}, and the analysis of the coded multicasting gain was provided in \cite{CodingforCaching}.

Besides caching, D2D communication underlaid with cellular networks has been intensively studied recently. In D2D communication, users can communicate directly without going through the base station. The advantages of D2D communications were studied in \cite{D2D_KB}, and it was shown that D2D communication could greatly enhance the spectral efficiency and lower the latency. A D2D-based heartbeat relaying framework was proposed to reduce signaling traffic and energy consumption in heartbeat transmission in \cite{jin2017reducing}. A comprehensive overview was provided in \cite{D2D_survey}, where different modeling assumptions and key considerations in D2D communications were detailed. In a D2D cellular network, users can choose to work in different modes. In cellular mode, users communicate through the base station just as cellular users; while in D2D mode, users communicate directly. Mode selection is a critical consideration in D2D communications, and many studies have been conducted in this area. For example, in \cite{mode_selection_DK}, mode selection problem was studied for a system with one D2D pair and one cellular user, and in \cite{Joint_selection}, a joint mode selection and resource allocation algorithm was proposed. Recently in \cite{li2016device}, mode selection and optimal resource allocation in D2D networks were studied under statistical queueing constraints.

In the literature, several studies have been performed to combine content caching with device-to-device (D2D) wireless networks. In such cases, a user can receive from its neighbors if these have cached the requested content. An overview on wireless D2D caching networks was provided in \cite{WirelessD2DCachingNetworks}, in which the key results for different D2D caching strategies were presented. To design caching policies for the wireless D2D network, the authors of \cite{MobileCachingD2DContentDelivery} proposed a caching policy that maximizes the probability that requests can be served via D2D communications. For a similar system setting, a caching policy that maximizes the average number of active D2D links was obtained in \cite{golrezaei2012wireless}. More recently, probabilistic caching policies have been applied in D2D caching problems. For instance, the authors in \cite{chen2017optimal} jointly optimize the probabilistic caching policy for D2D users and scheduling policies to maximize the successful offloading probability, and show remarkable improvement in offloading gain. In \cite{chen2017probabilistic}, an alternative optimization approach was proposed for the probabilistic caching placement aiming at maximizing the cache hit rate and cache-aided throughput. As another line of work, comparison between D2D caching and small cell caching was provided in \cite{chen2016d2d} and how the user density and content popularity distribution influence the caching performance was studied. This work was based on stochastic geometry models, in which nodes/users were distributed randomly.
We note that many works only tackle a simple case in which users have identical popularity vectors which is also not practical. In this paper, to address the D2D caching problem considering the channel fading and with unknown content popularity distribution, we design a caching algorithm that minimizes the average delay of the network, and also introduce an algorithm to learn the arrival intensity of the file requests. Since the users' requests may not belong to a particular family of distributions, a nonparametric estimation method \cite{tanner2001statistics} is needed in general. Most commonly used nonparametric estimators, including histogram, splines, wavelets, and kernel density estimator, are introduced in \cite{fan1996local}. Histogram is the oldest and most basic density estimator, which counts the number of samples falling into each bin \cite{silverman1986density}. The problems associated with this method are that the different choices of the size and initial points of bins have significant influence on the histogram, and the histogram is not smooth. In order to approximate smooth functions, spline functions are used. Rather than bins, the samples are fitted into a set of spline basis functions \cite{eubank1999nonparametric}. In \cite{donoho1995wavelet}, a wavelet based density estimation was introduced to achieve fast computation. In our work, we use the kernel density estimator for the intensity estimation. Kernel density estimator was first proposed in \cite{rosenblatt1956remarks} to provide smooth density functions. We introduce the kernel density function and describe our approach and design in detail in Section \ref{Sec:int_est}.

Our main contributions in this paper can be listed as follows:
\begin{enumerate}
\item We present an optimized kernel density estimator for the caching problem, which makes the proposed algorithm a model-free algorithm.
\item We provide a characterization of the average delay in both cellular and D2D modes, and our caching algorithm has the goal to minimize the average delay of the system, which is a critical quality of service requirement in especially delay-sensitive applications.
\item We propose an efficient and robust algorithm to solve the delay minimization problem. Our algorithm is applicable in settings with very general popularity models, with no assumptions on how file popularity varies among different users.
\item We further extend our algorithm to a more general setting, in which the system parameters and the distributions of channel fading change over time.
\end{enumerate}

The remainder of the paper is organized as follows. System model is described in Section \ref{Sec:formulation} and the problem formulation is provided in Section \ref{Sec:problem}. In Section \ref{Sec:int_est}, we present the intensity estimation algorithm while delay-aware caching algorithm is developed and described in Section \ref{Sec:algorithm}. Extensions with broadcasting in the transmission phase are discussed in Section \ref{Sec:broadcast}. Finally, numerical results are given in Section \ref{Sec:numerical} and conclusions are drawn in Section \ref{Sec:conclusion}.


\section{System Model}\label{Sec:formulation}

\subsection{System Model and Channel Allocation} \label{sec:model}
\begin{figure}
	\centering
	 \includegraphics[width=0.45\textwidth]{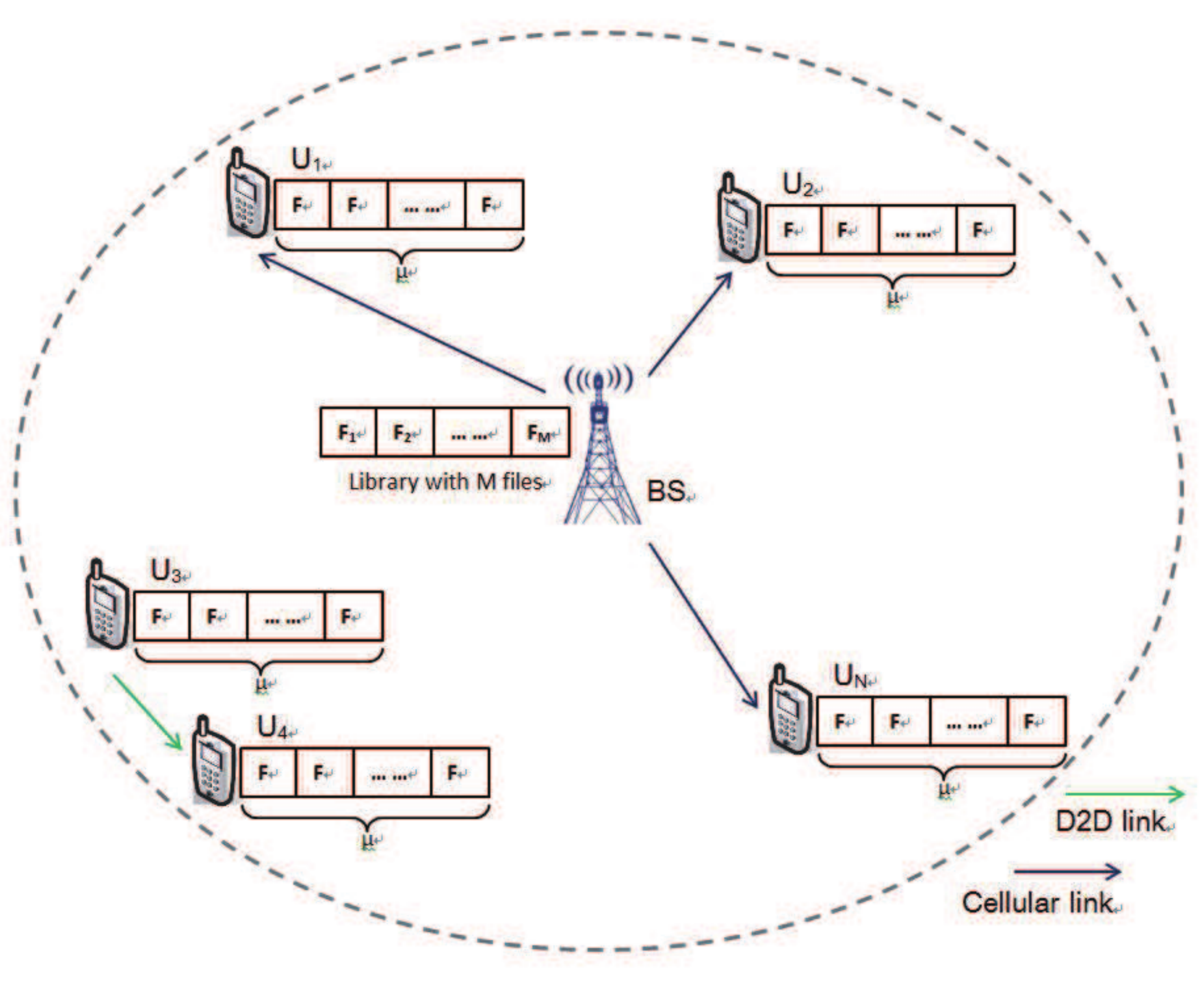}
	\caption{System model of a D2D cellular network with caches}\label{fig:system}
\end{figure}

As shown in Figure \ref{fig:system}, we consider a cellular network with one base station ($\BS$), in which a library of $M$ files ($\mathcal{F}_1$, $\mathcal{F}_2$, $\cdots$, $\mathcal{F}_M$) is stored, and we assume that the size of each file is fixed to $F$ bits \footnote{In the literature, it is noted that the base station may only store a portion of the library contents, and needs to acquire the remaining files from the content server \cite{WL_caching_overview}. Since we focus on the wireless transmission delay, we do not explicitly address the link between the base station and content server. Also, if the content files do not have the same size, we can further divide them into sub-files with equal size.}. There are $N$ users ($U_1$, $U_2$, $\cdots$, $U_N$) in the network who seek to get the content files from the library. Each user is equipped with a cache of size $\mu F$ bits, and therefore can store $\mu$ content files. The caching state is described by an $N\times M$ matrix $\mathbf{\Phi}$, whose $(i,j)$-th component has a value of $\phi_{i,j}=1$ if file $\mathcal{F}_j$ is cached at user $U_i$, and $\phi_{i,j}=0$ when the user $U_i$ does not have file $\mathcal{F}_j$ in its cache.

A discrete-time system is considered in this work, and the duration of each time frame is fixed at $T_0$. Throughout this work, we use $\kappa$ as the index of time frame, and denote $t_{\kappa}$ as the time instant when the $\kappa^{\tth}$ time frame begins. In practice, the contents on the cloud server are updated periodically, and the behavior of each user also follows a certain periodic pattern related to the update of the contents. For instance, the news websites may update several times a day, and the users often visit these news websites everyday following their update. Here we assume that the behavior of each user follows a periodic pattern with a period of $T_p$ time frames. More specifically, $\lambda_{i,j}(t)$, which represents the intensity function of the requests generated by user $U_i$ for file $\mathcal{F}_j$, is a continuous function of time with a period of $T_pT_0$.

It is assumed that the system also operates in a periodic pattern, whose period is far less than the period of the intensity function of the requests, so that the system can always keep updating as the users' request intensities are changing over time. At the beginning of each operation cycle, the base station collects information from the users, including the caching states, locations, maximum transmission powers, and estimated average numbers of requests, and runs the caching algorithm to determine the contents cached at each user. Then, the users update their caches and the system starts serving the requests. An operation period ends as soon as the base station starts running the caching algorithm again. We use $\upsilon$ as the index of operation cycle, and denote the duration of the $\upsilon^{\tth}$ operation cycle as $\tau_{\upsilon}$ time frames. With this notation, we can express the average number of requests generated by user $U_i$ requesting file $\mathcal{F}_j$ in the $\upsilon^{\tth}$ operation cycle as
\begin{align}\label{eq:nreq_ij}
\E\left\{\Nreq^{i,j}[\upsilon]\right\}=\int_{t_{\upsilon}}^{t_{\upsilon}+T_0\tau_{\upsilon}}\lambda_{i,j}(t)dt,
\end{align}
where $t_{\upsilon}$ represents the time instant in the continuous time domain at which the $\upsilon^{\tth}$ operation cycle starts. Then, the average number of requests generated by the $i^{\tth}$ user is
\begin{align}\label{eq:nreq_i}
\E\left\{\Nreq^{i}[\upsilon]\right\}=\sum_{j=1}^{M}\E\left\{\Nreq^{i,j}[\upsilon]\right\},
\end{align}
and the average number of requests received at the base station in the $\upsilon^{\tth}$ operation cycle is
\begin{align}\label{eq:nreq}
\E\left\{\Nreq[\upsilon]\right\}=\sum_{i=1}^{N} \E\left\{\Nreq^{i}[\upsilon]\right\}=\sum_{i=1}^{N}\sum_{j=1}^{M}\E\left\{\Nreq^{i,j}[\upsilon]\right\}.
\end{align}
On the device of each user, it is assumed that a certain learning algorithm is installed in order to learn the user behavior. In other words, the request intensity functions $\lambda_{i,j}(t)$ are estimated at each user, and the average number of requests generated in each operation cycle can be evaluated based on the intensity estimation. The base station collects the estimates of $\E\left\{\Nreq^{i,j}[\upsilon]\right\}$ from each user through a special channel at the beginning of the $\upsilon^{\tth}$ operation cycle.

In a D2D-enabled wireless network, users can choose to transmit in cellular mode or D2D mode. In the cellular mode, users request and receive information from the base station, while in the D2D mode, a user requests and receives information from another user through a D2D direct link. In our model, the users first check their local cache when a file is requested. If the user does not have the corresponding file in its own cache, it sends a request to the base station. We assume that the base station has knowledge of all fading \emph{distributions} (i.e., only has statistical information regarding the channels) and the cached files at each user. After receiving the request, the base station identifies the source node from which the file request can be served and allocates channel resources to the corresponding user. Therefore, the result of mode selection is determined by the result of source selection. If the source node is another user, then the requested file is sent over the direct D2D link and hence the communication is in D2D mode, otherwise the receiving user works in cellular mode and receives files from the base station. In source selection, among all the nodes (including the base station) who have the requested file, the node with the lowest average transmission delay to the receiver is selected as the transmitter. If the $i^{\text{th}}$ user is selected as a D2D transmitter, its maximum transmission power is denoted by $P_i^{\upsilon}$ in the $\upsilon^{\tth}$ operation cycle. The base station can serve multiple requests simultaneously using different channels, and its maximum transmission power is fixed at $P_b$ for each request.

In this work, we consider an OFDMA system with $N_c$ orthogonal channels, and the bandwidth of each channel is $B$. We assume that the background noise samples follow independent and identically distributed (i.i.d.) circularly-symmetric complex Gaussian distribution with zero mean and variance $\sigma^2$ at all receivers in all frequency bands, and the fading coefficients of the same transmission link are i.i.d. in different frequency bands. The fading coefficients are assumed to stay constant within one time frame, and change across different time frames. In this work, only the distributions of the fading coefficients are required at the base station, which mainly depend on the environment and the location of each user. A centralized computation scheme is used, and the base station sends the results of caching and scheduling algorithms to the users through additional control channels \footnote{The downlink and uplink control channels are available in 4G, LTE and 5G new radio (see e.g., \cite{ghosh2007uplink, wang2017uplink, love2008downlink,miao2017physical, UL17, DL17})}.
The average delay of the system depends on the resource allocation algorithm. In order to characterize the average delay of the system with caches at users, we provide a delay analysis in Section \ref{Sec:delay} below for a class of simple scheduling algorithms. We summarize the resource allocation assumptions for the delay evaluation as follows:
\begin{enumerate}
  \item Each channel can be used for the transmission of one requested file at most, and the transmission of a file cannot occupy multiple channels.
  \item D2D transmitters are not allowed to transmit to multiple receivers simultaneously. In other words, the file requests whose best source node is a user who is already transmitting cannot be assigned a channel resource by the base station.
  \item After a request is served, the corresponding transmitter keeps silent in the remaining time block, and the base station allocates the channel resource to other requests at the beginning of the following time block.
  \item All requests can be served in the same operation cycle, in which they are generated.
\end{enumerate}

Note that our intensity estimation and caching algorithms proposed in this work do not rely on these assumptions. Only the delay analysis in Section \ref{Sec:delay} requires them. The first three assumptions describe a class of simple scheduling algorithms, in which only point-to-point transmissions without spectrum reusing are considered. At the beginning of each time frame, base station assigns available channels to the requests, and each transmission link gets one channel at most, and uses the assigned channel exclusively. The transmitter transmits until the request is served and then releases the channel resource. For this type of scheduling algorithm, we provide average delay characterizations in Section \ref{Sec:delay}. In practice, more complicated scheduling algorithms might be used, and the delay analysis given in Section \ref{Sec:delay} needs to be updated in such a situation. However, the intensity estimation and caching algorithms are still valid if we can estimate the average delay between users. In that case, we need to estimate the average delay between users through simulation or learning methods, before calling the caching algorithm.

The last assumption described a light traffic load situation. In our setting, we can predict the average arrival number of the requests in an operation cycle via intensity estimation. However, the average packet delay should depend on the average number of requests that are served in an operation cycle, which is difficult to obtain, especially when complicated scheduling algorithms are employed. With this assumption, it is straightforward to have that the average number of requests that are served is equal to the average number of generated requests in an operation cycle. For a more general setting without this light traffic load assumption, the system should have a method to estimate the average number of requests that are served.

\subsection{Transmission Delay}\label{Sec:delay}
In this work, we use the transmission delay, which is defined as the number of time frames used to transmit a content file, as the performance metric. With this metric, the shortened transmission delay is guaranteed to reduce the waiting time for a request to be allocated a transmission resource. From the above discussion, the instantaneous channel capacity a transmission link in the $\kappa^{\text{th}}$ time frame is
\begin{align}
C[\kappa]=B\log_2\left(1+\frac{P_t}{B\sigma^2}z_\kappa\right)\hspace{0.4cm} \text{bits}/\text{s}
\end{align}
where $P_t$ is the transmission power, and $z_\kappa$ is the magnitude square of the corresponding fading coefficient in the $\kappa^{\text{th}}$ time frame. In order to maximize the transmission rate, all transmitters transmit at the maximum power level. Therefore,
\begin{align}
P_t=
\begin{cases}
P_b\hspace{1.3cm} \text{if the transmitter is the base station}\\
P_i^{\upsilon_\kappa}\hspace{1cm} \text{if the transmitter is the $i^{\text{th}}$ user}
\end{cases},
\end{align}
where $\upsilon_\kappa$ represents the index of operation cycle that contains the $\kappa^{\tth}$ time frame. When the four assumptions described in Section \ref{sec:model} hold, the duration to send a file is
\begin{align}
T=\min\left\{\tilde{t}:F\leq\sum_{\kappa=1}^{\tilde{t}} T_0 C[\kappa]\right\} \label{eq:transmissiondelay}
\end{align}
where $F$ is the size of each file, $T_0$ is the duration of each time frame, and $C[\kappa]$ is the instantaneous channel capacity in the $\kappa^{\text{th}}$ time frame. When the fading distribution is available, the average transmission delay of the link $U_i-U_j$, which is denoted by $\E\{T_{i,j}\}$, can be obtained through numerical methods or Monte-Carlo simulations. These average delay values can be stored in an $N\times N$ symmetric matrix $\bold{T_{\text{avg}}}$, whose component on the $i^{\text{th}}$ row and $j^{\text{th}}$ column is given by $\E\{T_{i,j}\}$ when $i\neq j$, and the diagonal element $\E\{T_{i,i}\}$ is the average delay between $U_i$ and the base station. According to our channel assumptions, the average delays of a transmission link are the same in every channel. Therefore, we only need to analyze the performance in a single channel.

The best source node of the request, which is generated by user $U_i$ requesting file $\mathcal{F}_j$, is the node which has file $\mathcal{F}_j$ and the smallest average transmission delay to $U_i$, and this minimum average delay is denoted by $D_{i,j}$ \footnote{If $U_i$ has cached $\mathcal{F}_j$, then the best source node is $U_i$ itself, and $D_{i,j}=0$.}. The best source of each possible request can be stored in an $N\times M$ table $\bold{S}$, in which each row corresponds to a user who generates the request, and each column corresponds to a file being requested. Also, these $D_{i,j}$ values can be collected in an $N\times M$ matrix $\bold{D}$. It is assumed that the matrices $\bold{T_{\text{avg}}}$ and $\bold{D}$ are constant within an operation cycle, and they are evaluated at the beginning of each operation cycle.

When a complicated scheduling algorithm is used, we need to estimate the average delay between users through simulation or learning methods in order to obtain the $\bold{T_{\text{avg}}}$ and $\bold{D}$ matrices. Once these two delay matrices are obtained, we can directly run our caching algorithm.

\section{Problem Formulation}\label{Sec:problem}
In the previous section, we have described the average delay matrices $\bold{T_{\text{avg}}}$ and $\bold{D}$. In this section, we formulate and discuss our caching problem. Using the $\bold{D}$ matrix, the average transmission delay of the requests generated in the $\upsilon^{\tth}$ operation cycle can be obtained as
\begin{align}
\eta[\upsilon] &=\frac{\sum_{i=1}^{N} \sum_{j=1}^{M} \E\left\{\Nreq^{i,j}[\upsilon]\right\} D^{\upsilon}_{i,j}}{\E\left\{\Nreq[\upsilon]\right\}} \notag \\
     &=\sum_{i=1}^{N} \sum_{j=1}^{M} \omega_{i,j}^{\upsilon}D^{\upsilon}_{i,j} \label{eq:eta}
\end{align}
where
\begin{align}\label{eq:w_ij}
\omega_{i,j}^{\upsilon}=\frac{\E\left\{\Nreq^{i,j}[\upsilon]\right\}}{\E\left\{\Nreq[\upsilon]\right\}}
\end{align}
represents the weight of the requests generated by user $U_i$ requesting file $\mathcal{F}_j$. In this work, our goal is to minimize the average content delay, and our caching problem is formulated as
\begin{align}
\textbf{P1:}\hspace{1cm}&\underset{\mathbf{\Phi}^{\upsilon}}{\text{Minimize}}\hspace{1.2cm} \eta[\upsilon]\\
&\text{Subject to} \hspace{0.7cm} \sum_{j=1}^{M} \phi_{i,j}^{\upsilon}=\mu \label{constraint_1}\\
&\hspace{2.4cm} \sum_{j=1}^{M} \bigg|\phi_{i,j}^{\upsilon}-\phi_{i,j}^{\upsilon-1}\bigg|\leq 2\xi_i^{\upsilon}\label{constraint_2}\\
&\hspace{2.4cm} \phi_{i,j}^{\upsilon}\in\{0,1\}
\end{align}
for the $\upsilon^{\tth}$ operation cycle, where $\mathbf{\Phi}^{\upsilon}$ is the caching result indicator matrix. The constraint in (\ref{constraint_1}) arises due to the maximum cache size. It is obvious that the optimal caching policy must use all caching space. In (\ref{constraint_2}), $\xi_i^{\upsilon}$ is the upper limit on the number of cache files that will be replaced in the current operation cycle. Due to requirements regarding energy efficiency and current traffic load, each user may be able to update only a few cache contents.

\section{Intensity Estimation Algorithm}\label{Sec:int_est}
In this section, we propose our intensity estimation algorithm for each $<$user,file$>$ pair. As mentioned in the previous section, the arrival intensity function $\lambda_{i,j}(t)$, which represents the arrival intensity function of the requests generated by user $U_i$ requesting file $\mathcal{F}_j$, has a period of $T_pT_0$ in the continuous time domain. Therefore, we only need to find an estimate $\hat{\lambda}_{i,j}(t)$ for $t\in [0,T_pT_0]$, and assume $\hat{\lambda}_{i,j}(t)=0$ for $t\in (-\infty,0)\bigcup(T_pT_0,+\infty)$.

\subsection{Kernel Density Estimation}
In the literature, kernel density estimator is a nonparametric estimator which can decrease the modeling biases \cite{rosenblatt1956remarks}. The density function is directly estimated from data samples observed in noise, without assuming the form of the real intensity function \cite{tanner2001statistics}. Suppose we have $N_p$ samples collected over $N_T$ periods, and we collect their relative arriving time \footnote{We set the starting time of each period as $0$, then the relative arriving time of each request is in the region $[0,T_pT_0]$.} into a set $\Xi=\{t_1,t_2,\cdots,t_{N_p}\}$. Then the kernel density estimator of the intensity function is given by \cite{berman1989estimating}
\begin{align}\label{eq:estimator}
\hat{\lambda}_{i,j}(t)=\frac{1}{N_T}\sum_{t_\alpha\in\Xi}\frac{1}{W}K\left(\frac{t-t_\alpha}{W}\right) \mathbf{1}\{0\leq t\leq T_pT_0\},
\end{align}
where $K(\cdot)$ is the kernel function, $\mathbf{1}\{\cdot\}$ is the indicator function, and $W$ represents the bandwidth of the kernel function (\emph{which is a different concept from the bandwidth of the channel}). In the literature, performances of different kernel functions have been studied \cite{wand1994kernel}. In this work, we choose the Epanechnikov function as the kernel function, which is expressed as
\begin{align}\label{eq:Ekov}
K(x)=
\begin{cases}
&0.75(1-x^2)\quad |x|\leq 1 \\
&0 \hspace{2.2cm} \text{otherwise}
\end{cases}.
\end{align}
The bandwidth $W$ in (\ref{eq:estimator}) controls the smoothness of the estimated intensity function. When $W$ is small, the kernel function is narrow, and it becomes flat as $W$ increases. In the literature, it was shown that there exists an optimal bandwidth that minimizes the integrated squared error (ISE) \cite{wand1994kernel}, and the cross-validation method was proposed in \cite{bowman1984alternative} for bandwidth selection, which is discussed in detail in Section \ref{sec:CV}. We should note that the intensity function given in (\ref{eq:estimator}) only works for the estimators defined over the interval $t \in (-\infty, +\infty)$. To have it work in our case with $t \in [0, T_pT_0]$,  end correction is needed to make up the effect near the lower and upper bounds. The corrected kernel function and the proof are provided in \ref{sec:EC} below. Specifically, in the following subsections, we first introduce the end correction and cross-validation for intensity estimation, and subsequently the intensity learning algorithm is provided in Section \ref{sec:sum_estimation}.

\subsection{End Correction}\label{sec:EC}
In this subsection, we propose a novel end correction method for our intensity estimation algorithm. In order to have accurate/unbiased estimation, the intensity estimator given in (\ref{eq:estimator}) should satisfy
\begin{align}\label{eq:int_req1}
\int_{0}^{T_pT_0}\hat{\lambda}_{i,j}(t) dt=\frac{N_p}{N_T}.
\end{align}
As the number of samples increases, this integral should convergence to the average number of requests arriving in a behavior period. For the Epanechnikov kernel function $K(\cdot)$ corresponding to a sample arriving at $t=t_\alpha$, we have
\begin{align}\label{eq:kernel_int}
\int_{-\infty}^{+\infty}\frac{1}{W}K\left(\frac{t-t_\alpha}{W}\right) dt=\int_{t_\alpha-W}^{t_\alpha+W}\frac{1}{W}K\left(\frac{t-t_\alpha}{W}\right) dt=1,
\end{align}
and $K(\frac{t-t_\alpha}{W})$ is non-zero only for $t\in [t_\alpha-W,t_\alpha+W]$. Then, for those samples with arriving time close to the ends of the interval $[0,T_pT_0]$, or more specifically when $t_\alpha<W$ or $T_pT_0-t_\alpha<W$, we have
\begin{align}
\int_{0}^{T_pT_0} \frac{1}{W}K\left(\frac{t-t_\alpha}{W}\right) dt<1,
\end{align}
which results in
\begin{align}
\int_{0}^{T_pT_0}\hat{\lambda}_{i,j}(t) dt  = \frac{1}{N_T} \int_{0}^{T_pT_0} \sum_{\alpha=1}^{N_p}\frac{1}{W}K\left(\frac{t-t_\alpha}{W}\right) dt  < \frac{N_p}{N_T}.
\end{align}

\begin{figure}
	\centering
	 \includegraphics[width=\figsize\linewidth]{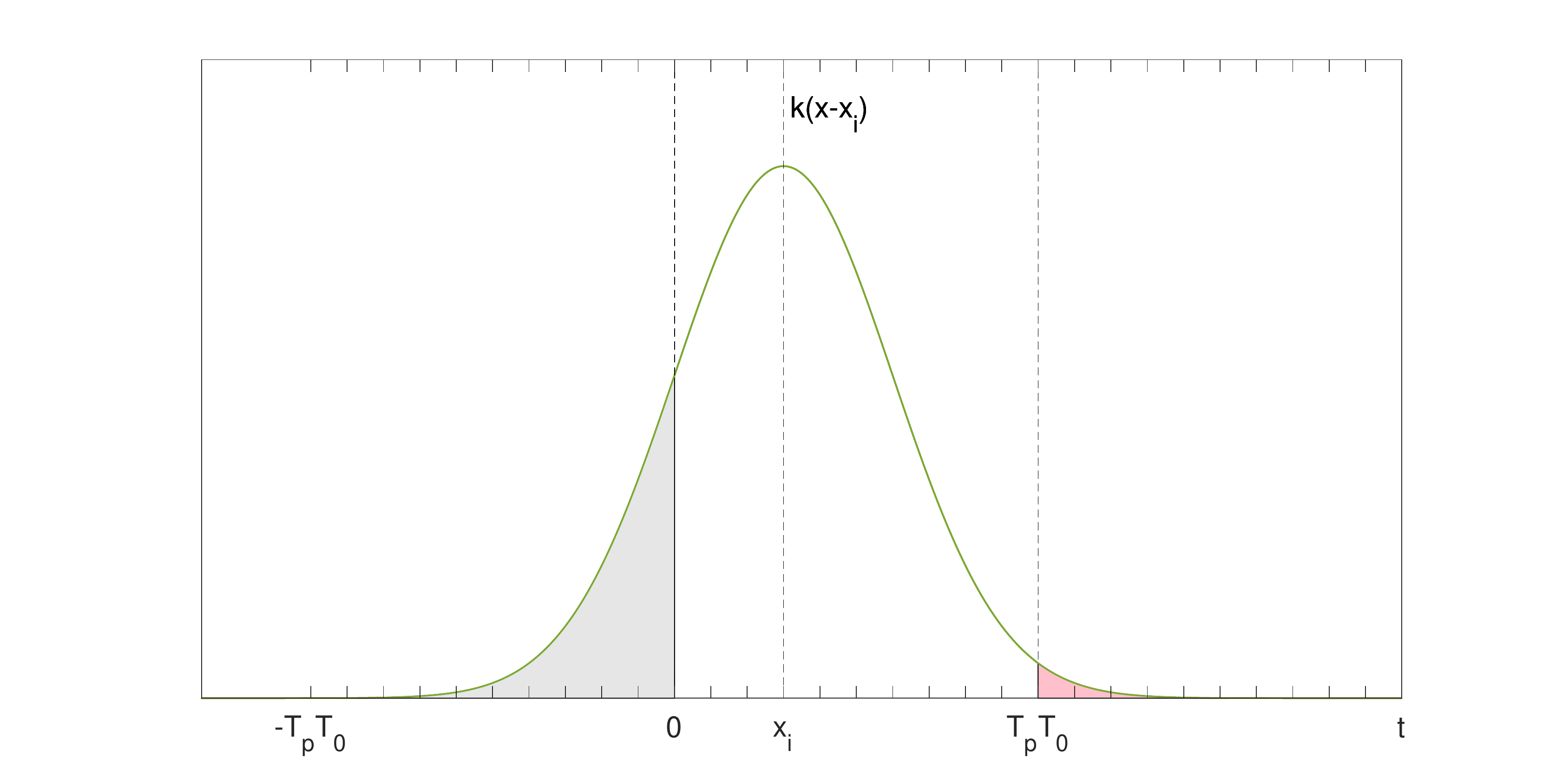}
	\caption{A kernel function exceeding beyond $t \in [0, T_pT_0]$}
	\label{fig:endbound}
\end{figure}

As shown in Fig. \ref{fig:endbound}, the kernel function can be beyond the range $t \in [0,T_pT_0]$, and the areas in gray and pink shadow are the losses or biases at left and right ends, respectively.

This bias can be fixed using end correction methods, in which a corrected kernel function $\tilde{K}(\cdot)$ is designed to satisfy
\begin{align}\label{eq:int_req2}
\int_{0}^{T_pT_0} \frac{1}{W}\tilde{K}\left(\frac{t-t_\alpha}{W}\right) dt=1,
\end{align}
which can guarantee the relationship in (\ref{eq:int_req1}) after replacing the kernel function in (\ref{eq:estimator}) with its corrected version.

In this work, we choose the corrected kernel function as
\begin{align}
\tilde{K}\left(\frac{t-t_\alpha}{W}\right)=&K\left(\frac{t-t_\alpha}{W}\right)+K\left(\frac{t-t_\alpha+T_pT_0}{W}\right)
\nonumber \\
&+K\left(\frac{t-t_\alpha-T_pT_0}{W}\right), \label{eq:kernel_correct}
\end{align}
where $K(\cdot)$ is again selected as the Epanechnikov function. The corrected kernel function incorporates two additional copies of the original kernel function and translates the copies one period to the right and leftm respectively.
Fig. \ref{fig:endcorrection} shows that the loss at left end is corrected by the copy of kernel function at the right side, and the loss at the right side is corrected by the copy at the left side.

\begin{figure}
	\centering
	 \includegraphics[width=\figsize\linewidth]{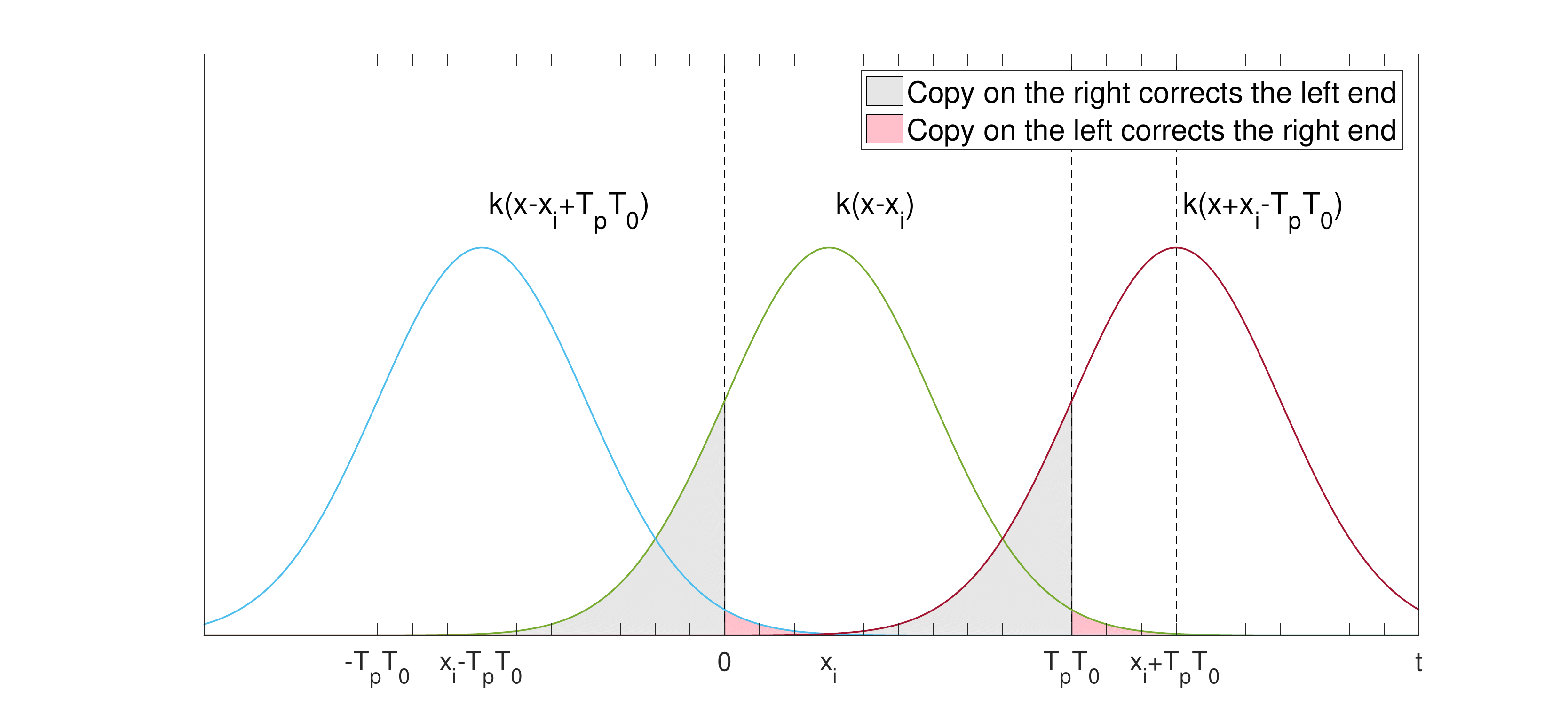}
	\caption{Corrected kernel function with end correction at both sides}
	\label{fig:endcorrection}
\end{figure}

\begin{Lem}\label{Lem1}
The corrected kernel function in (\ref{eq:kernel_correct}) satisfies the requirement described in (\ref{eq:int_req2}) when the bandwidth $W$ is smaller than $T_pT_0$.
\end{Lem}
\begin{proof}
We first insert the kernel function in  (\ref{eq:kernel_correct}) into the integral in (\ref{eq:int_req2}), and obtain (\ref{eq:proof_1}) given at the top of the next page.
\begin{figure*}
\begin{align}
& \int_{0}^{T_pT_0} \frac{1}{W}\tilde{K}\left(\frac{t-t_\alpha}{W}\right) dt \notag\\
& = \int_{0}^{T_pT_0} \frac{1}{W} \left\{K\left(\frac{t-t_\alpha}{W}\right)+K\left(\frac{t-t_\alpha+T_pT_0}{W}\right)+K\left(\frac{t-t_\alpha-T_pT_0}{W}\right)\right\} dt\\
& = \int_{0}^{T_pT_0} \frac{1}{W} K\left(\frac{t-t_\alpha}{W}\right)dt + \int_{0}^{T_pT_0} \frac{1}{W} K\left(\frac{t-t_\alpha+T_pT_0}{W}\right)dt + \int_{0}^{T_pT_0} \frac{1}{W} K\left(\frac{t-t_\alpha-T_pT_0}{W}\right)dt\label{eq:proof_1}
\end{align}
\end{figure*}
According to (\ref{eq:Ekov}), $K\left(\frac{t-t_\alpha}{W}\right)$ is non-zero only for $t\in [t_\alpha-W,t_\alpha+W]$. Therefore, when $W<T_pT_0$, we have
\begin{align}
&\int_{0}^{T_pT_0} \frac{1}{W} K\left(\frac{t-t_\alpha+T_pT_0}{W}\right)dt \notag
\\
&=\int_{0}^{\max\{0,W+t_\alpha-T_pT_0\}} \frac{1}{W} K\left(\frac{t-t_\alpha+T_pT_0}{W}\right)dt \label{eq:proof_2}\\
&=\int_{T_pT_0}^{\max\{T_pT_0,W+t_\alpha\}} \frac{1}{W} K\left(\frac{\hat{t}-t_\alpha}{W}\right)d\hat{t}. \label{eq:proof_3}
\end{align}
From (\ref{eq:proof_2}) to (\ref{eq:proof_3}), we replace $t$ with a new variable $\hat{t}=t+T_pT_0$. Similarly, we can also have
\begin{align}
&\int_{0}^{T_pT_0} \frac{1}{W} K\left(\frac{t-t_\alpha-T_pT_0}{W}\right)dt \nonumber
\\
&= \int_{\min\{T_pT_0+t_\alpha-W,T_pT_0\}}^{T_pT_0} \frac{1}{W} K\left(\frac{t-t_\alpha-T_pT_0}{W}\right)dt \label{eq:proof_5}\\
& = \int_{\min\{t_\alpha-W,0\}}^{0} \frac{1}{W} K\left(\frac{\hat{t}-t_\alpha}{W}\right)d\hat{t}. \label{eq:proof_6}
\end{align}
From (\ref{eq:proof_5}) to (\ref{eq:proof_6}), we replace $t$ with a new variable $\hat{t}=t-T_pT_0$.

Plugging (\ref{eq:proof_3}) and (\ref{eq:proof_6}) into (\ref{eq:proof_1}), we can get
\begin{align}
& \int_{0}^{T_pT_0} \frac{1}{W}\tilde{K}\left(\frac{t-t_\alpha}{W}\right) dt \notag\\
& = \int_{\min\{0,t_\alpha-W\}}^{0} \frac{1}{W} K\left(\frac{t-t_\alpha}{W}\right)dt \notag\\
& \hspace{.3cm}+ \int_{0}^{T_pT_0} \frac{1}{W} K\left(\frac{t-t_\alpha}{W}\right)dt \notag
\\
& \hspace{.3cm}+ \int_{T_pT_0}^{\max\{t_\alpha+W,T_pT_0\}} \frac{1}{W} K\left(\frac{t-t_\alpha}{W}\right)dt\\
& = \int_{\min\{0,t_\alpha-W\}}^{\max\{t_\alpha+W,T_pT_0\}} \frac{1}{W} K\left(\frac{t-t_\alpha}{W}\right)dt \\
& = \int_{t_\alpha-W}^{t_\alpha+W} \frac{1}{W} K\left(\frac{t-t_\alpha}{W}\right)dt \\
& = 1
\end{align}

and Theorem \ref{Lem1} is proved.
\end{proof}

\begin{figure}
	\centering
	 \includegraphics[width=\figsize\linewidth]{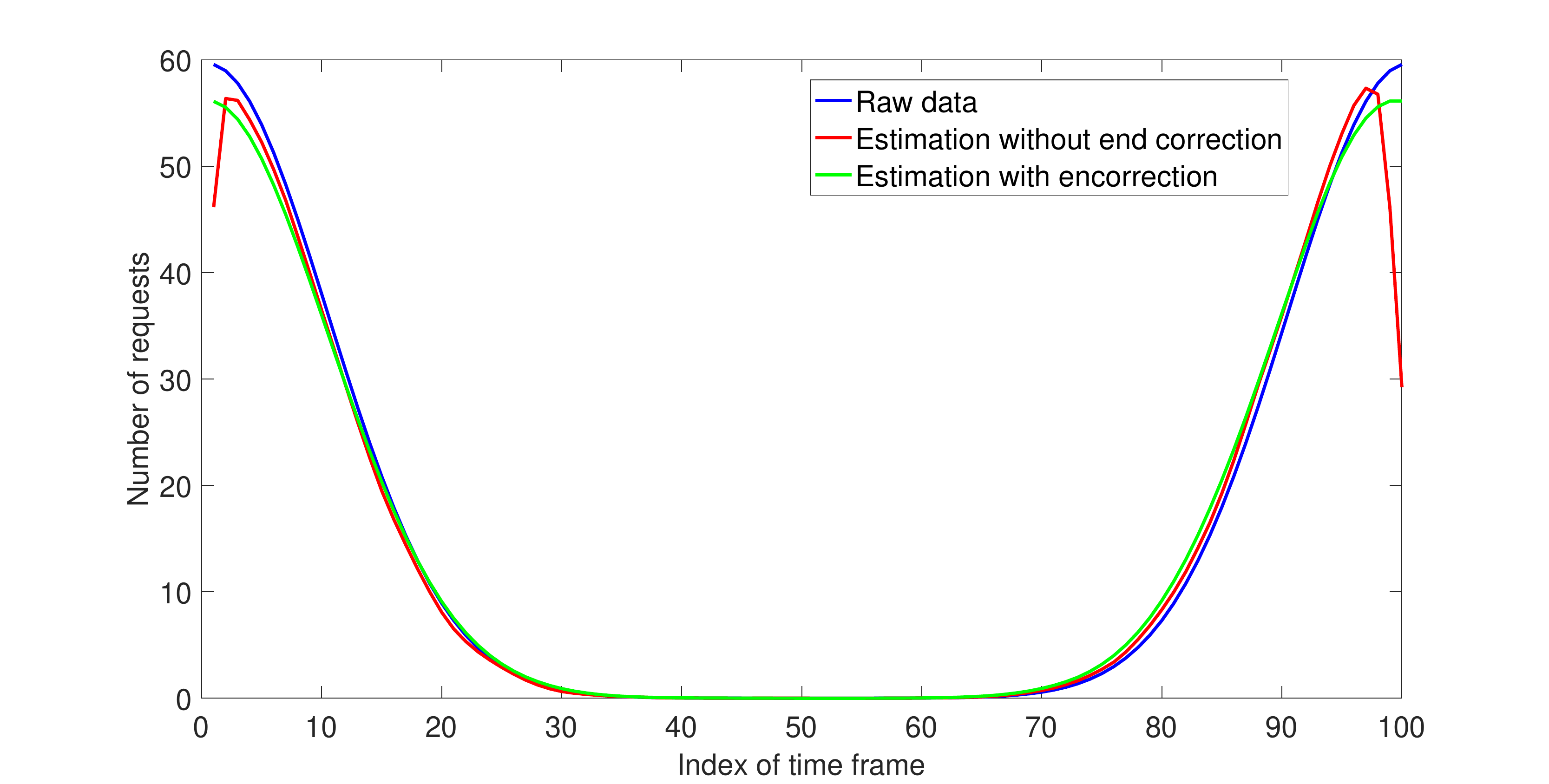}
	\caption{Estimation results with and without with end correction}
	\label{fig:endcorrectioncompare}
\end{figure}

Therefore, a unbiased intensity estimator can be obtained from combining the results in (\ref{eq:kernel_correct}) and (\ref{eq:estimator}), which can be expressed as
\begin{align}
\hat{\lambda}_{i,j}(t)&=\frac{1}{N_T}\sum_{t_\alpha\in\Xi}\frac{1}{W}\tilde{K}\left(\frac{t-t_\alpha}{W}\right) \mathbf{1}\{0\leq t\leq T_pT_0\}\\
&= \frac{1}{N_T}\sum_{t_\alpha\in\Xi}\frac{1}{W}\left\{ K\left(\frac{t-t_\alpha}{W}\right)+K\left(\frac{t-t_\alpha+T_pT_0}{W}\right)\right.\notag
\\
&\hspace{1.7cm}\left.+K\left(\frac{t-t_\alpha-T_pT_0}{W}\right) \right\}\mathbf{1}\{0\leq t\leq T_pT_0\}, \label{eq:estimator_correct}
\end{align}


for $W<T_pT_0$.

Fig.\ref{fig:endcorrectioncompare} shows the comparison between the estimation results obtained form the original kernel function and the corrected kernel function. We notice that the performance of the original kernel function drops rapidly near the bounds, while the corrected kernel function performs relatively stable near the bounds. In the following subsection, we discuss the bandwidth optimization.

\subsection{Cross-Validation}\label{sec:CV}
In this subsection, we study bandwidth optimization via the cross-validation method. This method provides an estimator of the ISE, and then form an ISE minimization problem to find the optimal bandwidth $W$. Using the analysis in \cite{bowman1984alternative}, we can obtain the ISE estimator given as
\begin{align}\label{eq:MISE}
\text{ISE}\approx \frac{N_T^2}{9N_p^2}\int_{0}^{T_pT_0} \hat{\lambda}_{i,j}^2(t) dt-\frac{2N_T}{3N_p(3N_p-1)}\sum_{\alpha=1}^{N_p} \hat{\lambda}_{i,j\;-\alpha}(t_\alpha)
\end{align}
where $\hat{\lambda}_{i,j\;-\alpha}$ is the intensity estimator when the $\alpha^{\text{th}}$ sample is removed from the sample set $\Xi$. Then the ISE minimization problem can be formulated as
\begin{align}\label{eq:P1}
&\underset{W}{\text{Minimize}}\hspace{2.5cm} \text{ISE} \\
&\text{Subject to}\hspace{2cm} 0<W<T_pT_0 \;.
\end{align}
This optimization problem can be solved easily because it is a minimization over a single scalar parameter on a bounded interval. With the optimal bandwidth, the intensity estimator $\hat{\lambda}_{i,j}$ can be determined.

\subsection{Summary of the Intensity Estimation Algorithm}\label{sec:sum_estimation}
In this subsection, we summarize the intensity estimation process. At the end of each behavior period, the intensity estimation program on the device of the $i^{\tth}$ user collects the requests arriving in the most recent $N_T$ periods, and form the sample set $\Xi$ for each content. Then, the program finds the optimal bandwidth $W$ by solving the ISE minimization problem, and determine the intensity estimators using (\ref{eq:estimator_correct}).

With the intensity estimators, each user can learn the intensity of the file requests and predict the average number of arriving requests using (\ref{eq:nreq_ij}), and then send this information to the base station at the beginning of each operation cycle. In the next section, we investigate the caching algorithm performed at the base station, which determines the contents cached at each user.

\section{Delay-Aware Caching}\label{Sec:algorithm}

In this section, we propose our caching algorithm that solves problem $\textbf{P1}$. Note that the objective in problem $\textbf{P1}$ is not convex, and the solution space is a discrete set with size $(\frac{M!}{(M-\mu)!\,\mu!})^N$. Therefore, the globally optimal solution can only be obtained via exhaustive search. In this work, we propose an efficient algorithm to determine a caching policy with delay performance close to the optimal solution.
Next, we describe the algorithm in detail.

\subsection{Caching Algorithm}
Our algorithm is a greedy algorithm, which searches over a subset of the solution space with smaller size.  At the initial point, we assume that all caches are empty, and every user has to operate in cellular mode, in which they only receive files from the base station. Then, in each step, we find the best $<$file,user$>$ pair, which provides the maximum delay improvement (or equivalently reduction in delay) if the selected file is stored in the cache of the corresponding user. We also assume that users do not physically clear their caches. They carry over their cache states, which now update only according to the final caching results, to next operation cycle. This process needs to be repeated $N\mu$ times, in order to fill all cache space, and the final caching policy is obtained.

\begin{table}
\centering
\caption{\label{Algorithm1}Algorithm \ref{Algorithm1}}
\begin{tabular}{p{8cm}}
\hline
\hline
Find the delay improvement for a $<$file,user$>$ pair\\
\hline
\hline
\textbf{Input :} user index $i$, file index $j$, caching indicator in the $\upsilon^{\tth}$ operation cycle $\phi_{i,j}^{\upsilon}$, caching indicator in the $(\upsilon-1)^{\tth}$ operation cycle $\phi_{i,j}^{\upsilon -1}$, weight matrix $\boldsymbol{\omega}^{\upsilon }=\{\omega_{i,j}^{\upsilon }\}$, source table $\bold{S}^{\upsilon }$, delay matrices $\bold{T_{\text{avg}}}^{\upsilon }$ and $\bold{D}^{\upsilon }$, the number of update vector at $\upsilon^{\tth}$ operation cycle  $ \{\xi_{1}^{\upsilon},\xi_{2}^{\upsilon},...,\xi_{N}^{\upsilon}\}$.\\
\textbf{Output :} delay improvement $g_{i,j}$, updated source table $\hat{\bold{S}}$, updated optimal delay matrix $\hat{\bold{D}}$, updated number of update vector $\{\hat{\xi_{1}}, \hat{\xi_{2}},...,\hat{\xi_{N}}\}$.\\
\hline
\textbf{Initialization :} $\hat{\bold{S}}=\bold{S}^{\upsilon }$,  $\hat{\bold{D}}=\bold{D}^{\upsilon }$ and $\hat{\xi_{i}}=\xi_{i}^{\upsilon}$\\

\textbf{If} $\phi_{i,j}^{\upsilon}=1$ or $\hat{\xi_{i}} \leq 0 $ \\
\hspace{0.5cm} $g_{i,j}=0$, end process.\\
\textbf{Else if} $\phi_{i,j}^{\upsilon -1 }=0$ and $\xi_{i}^{\upsilon } \leq 0 $\\
\hspace{0.5cm} $g_{i,j}=0$, end process.\\
\textbf{Else}\\
\hspace{0.5cm} $g_{i,j}=\omega_{i,j}^{\upsilon } D_{i,j}^{\upsilon }$ and update $\hat{S}_{i,j}\leftarrow U_i$, $\hat{D}_{i,j}=0$, $\hat{\xi_{i}}=\xi_{i}^{\upsilon} -1$.\\
\hspace{0.5cm} \textbf{For} $k=1:N$\\
\hspace{1cm} \textbf{If} $D_{k,j}^{\upsilon }>T_{i,k}^{\upsilon }$ and $i\neq k$\\
\hspace{1.5cm} $g_{i,j}=g_{i,j}+\omega_{k,j}^{\upsilon } (D_{k,j}^{\upsilon }-T_{i,k}^{\upsilon })$ \\
\hspace{1.5cm} update $\hat{D}_{k,j}=T_{i,k}^{\upsilon }$ and $\hat{S}_{k,j}\leftarrow U_i$\\
\hspace{1cm} \textbf{End}\\
\hspace{0.5cm} \textbf{End}\\
\textbf{End}\\

\hline\hline
\end{tabular}
\end{table}

In Table \ref{Algorithm1}, we describe Algorithm \ref{Algorithm1} in detail, which calculates the delay improvement and determines the updated $\hat{\bold{S}}$ and $\hat{\bold{D}}$ matrices accordingly when we cache file $\mathcal{F}_j$ at user $U_i$. First, we check if $\mathcal{F}_j$ has already been cached at $U_i$. If file $\mathcal{F}_j$ is cached at user $U_i$ in the current operation cycle, we end the process, and return the delay improvement $g_{i,j}=0$; if not, we continue to check if the file $\mathcal{F}_j$ was cached at user $U_i$ in the previous operation cycle and if the number of updates of user $U_i$ in the operation cycle is within the upper bound $\xi_i^{\upsilon}$. If file $\mathcal{F}_j$ was not cached at user $U_i$ in the previous operation cycle and the limit on the number of updates does not allow user $U_i$ to update in the current operation cycle,  we end the process, and return the delay improvement $g_{i,j}=0$; if not, we set $g_{i,j}=\omega_{i,j}^{\upsilon } D_{i,j}^{\upsilon }$ because that is the reduction in $\eta$ at user $U_i$ if it adds $\mathcal{F}_j$ to its cache. And at the same time, we update the available number of updates at user $U_i$ in the current operation cycle. Then, we need to sum up all reductions at each user. At user $U_k$, if $D_{k,j}^{\upsilon }>T_{i,k}^{\upsilon }$, then D2D link $U_i-U_k$ has the lowest average delay for $U_k$ to receive $\mathcal{F}_j$ and the reduction at $U_k$ is $\omega_{k,j}^{\upsilon } (D_{k,j}^{\upsilon }-T_{i,k}^{\upsilon })$; if not, then caching $\mathcal{F}_j$ at $U_i$ does not help to improve the delay performance at $U_k$.

\begin{table}
\centering
\caption{\label{Algorithm2}Algorithm \ref{Algorithm2}}
\begin{tabular}{p{8cm}}
\hline
\hline
Find the optimal $<$file,user$>$ pair to be added in the updated caching result, leading to maximum delay improvement\\
\hline
\hline
\textbf{Input :} weight matrix $\boldsymbol{\omega}^{\upsilon }=\{\omega^{\upsilon }_{i,j}\}$, caching indicator matrix in the $\upsilon^{\tth}$ operation cycle  $\mathbf{\Phi}^{\upsilon}$, caching indicator matrix in the $(\upsilon -1)^{\tth}$ operation cycle  $\mathbf{\Phi}^{\upsilon -1}$, source table $\bold{S}^{\upsilon }$, delay matrices $\bold{T_{\text{avg}}}^{\upsilon }$ and $\bold{D}^{\upsilon }$, index of time frame $\kappa$, the number of update vector at the $\upsilon^{\tth}$ operation cycle $ \{\xi_{1}^{\upsilon},\xi_{2}^{\upsilon},...,\xi_{N}^{\upsilon}\}$.\\
\textbf{Output :} new source table $\bold{S}^{\upsilon}$, new optimal delay matrix $\bold{D}^{\upsilon}$, new caching indicator matrix $\mathbf{\Phi}^{\upsilon}$.\\
\hline
\textbf{Initialization :} set optimal delay improvement $g^*=0$, and set the corresponding $\bold{S}^*=\bold{S}^{\upsilon }$, $\bold{D}^*=\bold{D}^{\upsilon }$, the updated number of update vector $\{\xi_{1}^*, \xi_{2}^*,...,\xi_{N}^*\} = \{\xi_{1}^{\upsilon},\xi_{2}^{\upsilon},...,\xi_{N}^{\upsilon}\}$.\\
\\
\textbf{For} $i=1:N$\\

\hspace{0.5cm} \textbf{For} $j=1:M$\\
\hspace{1cm} run Algorithm \ref{Algorithm1} for $<U_i,\mathcal{F}_j>$, to obtain\\
\hspace{1cm} the gain $g_{i,j}$ and the corresponding $\hat{\bold{S}}$ and $\hat{\bold{D}}$.\\
\hspace{1cm} \textbf{If} $g_{i,j}>g^*$\\
\hspace{1.5cm} update $g^*=g_{i,j}$, $\{\xi_{1}^*, S^* = \hat{S}, D* = \hat{D}$,\\
\hspace{1.5cm} $\widetilde{i}=i$, and $\widetilde{j}=j$.\\
\hspace{1cm} \textbf{End}\\
\hspace{0.5cm} \textbf{End}\\
\textbf{End}\\

$\phi_{\widetilde{i},\widetilde{j}}^{\upsilon}=1$, $\xi_{\widetilde{i}}^* = \xi_{\widetilde{i}}^{\upsilon} -1$,\\

update   $\mathbf{\Phi}^{\upsilon}$, $\bold{S}^{\upsilon}=\bold{S}^*$, $\bold{D}^{\upsilon}=\bold{D}^*$, and $\{\xi_{1}^{\upsilon},\xi_{2}^{\upsilon},...,\xi_{N}^{\upsilon}\}=\{\xi_{1}^*, \xi_{2}^*,...,\xi_{N}^*\}$.
\\
\hline\hline
\end{tabular}
\end{table}

Based on Algorithm \ref{Algorithm1}, Algorithm \ref{Algorithm2} described in Table \ref{Algorithm2} helps to find the optimal $<$file,user$>$ pair to be added to the updated caching result, which leads to the maximum delay reduction. In Algorithm \ref{Algorithm2}, $\widetilde{i}$ and $\widetilde{j}$ record the optimal user index and file index, respectively. $g^*$ tracks the maximum delay improvement, $\{\xi_{1}^*, \xi_{2}^*,...,\xi_{N}^*\}$ tracks the number of updates at all users, and $\bold{S}^*$ and $\bold{D}^*$ record the new source table and minimum delay matrix, respectively, after caching $\mathcal{F}_{\widetilde{j}}$ at $U_{\widetilde{i}}$. We search over all $NM$ possible $<$file,user$>$ combinations, find their delay improvements and update $g^*$, $\widetilde{i}$, $\widetilde{j}$, $\bold{S}^*$ and $\bold{D}^*$ accordingly. At user $U_i$, we check if there is available space in its cache. If its cache is full, we directly jump to the next user $U_{i+1}$. For each $<$file,user$>$ pair, we run Algorithm \ref{Algorithm1} to calculate the corresponding delay improvement, and compare it with $g^*$. If a $<$file,user$>$ pair exceeds the maximum delay improvement up to that point, we perform the update accordingly. Every time we run Algorithm \ref{Algorithm2}, we cache one more file at a user. Therefore, we need to run Algorithm \ref{Algorithm2} $N\mu$ times to obtain the final caching result, and this process is described in Algorithm \ref{Algorithm3} in Table \ref{Algorithm3}.

\begin{table}
\centering
\caption{\label{Algorithm3}Algorithm \ref{Algorithm3}}
\begin{tabular}{p{8cm}}
\hline
\hline
Caching Algorithm for the $\upsilon^{\tth}$ operation cycle\\
\hline
\hline
\textbf{Input :}  weight matrix $\boldsymbol{\omega}^{\upsilon }=\{\omega_{i,j}^{\upsilon }\}$, and delay matrix $\bold{T_{\text{avg}}}^{\upsilon }$,caching indicator matrix in the $\upsilon^{\tth}$ operation cycle  $\mathbf{\Phi}^{\upsilon}$, caching indicator matrix in the $(\upsilon -1)^{\tth}$ operation cycle  $\mathbf{\Phi}^{\upsilon -1}$, number of update vector $ \{\xi_{1}^{\upsilon},\xi_{2}^{\upsilon},...,\xi_{N}^{\upsilon}\}$.\\
\textbf{Output :} updated caching indicator matrix $\mathbf{\Phi}^{\upsilon}$, source table $\bold{S}^{\upsilon}$.\\
\hline
\textbf{Initialization :} for all requests, $S_{i,j}^{\upsilon }\leftarrow \BS$, $D_{i,j}^{\upsilon }=T_{i,i}$. Set all $\phi_{i,j}^{\upsilon}=0$.\\
\textbf{For} $loop=1:N\mu$\\
\hspace{0.5cm} run Algorithm \ref{Algorithm2} to cache a file and update the result.\\
\textbf{End}\\
\hline\hline
\end{tabular}
\end{table}

For our proposed caching algorithm, we initially have all caches empty, and all users work in cellular mode, in which they only receive files from the base station at first. We assume that the system has calculated the average delay between each pair of nodes, and stored the delay matrix $\bold{T_{\text{avg}}}^{\upsilon }$ at the base station. Then, the base station runs Algorithm \ref{Algorithm2} $N\mu$ times, and in each time we cache one more file and update the caching indicator $\mathbf{\Phi}^{\upsilon}$, source table $\bold{S}^{\upsilon}$, and minimum delay matrix $\bold{D}^{\upsilon}$ accordingly. Finally, the base station sends the caching files to the users when the traffic load is low.

The solution of \textbf{P1} is described below:
\begin{enumerate}
	\item At the beginning of the $\upsilon^{\text{th}}$ cycle, the system estimates the delay matrix $\bold{T_{\text{avg}}}^{\upsilon-1}$, and weight vector $\boldsymbol{\omega}^{\upsilon-1}$ according to the request intensity estimation in the previous cycle. The base station receives the transmission powers $P_i^\upsilon$ from the users, determines the cycle period $\tau^\upsilon$ and the upper bound $\xi_i^{\upsilon}$, and then predicts $\bold{T_{\text{avg}}}^{\upsilon}$, and $\boldsymbol{\omega}^{\upsilon}$.
	\item Algorithm \ref{Algorithm2} is repeated $N\mu$ times to determine the caching result in the $\upsilon^{\text{th}}$ cycle.
\end{enumerate}
After this process, the base station sends the cache contents to each user, and conducts regular transmissions after updating the cache files at each user.

As we have mentioned in Section \ref{Sec:formulation}, this proposed algorithm does not require the first $4$ resource allocation assumptions described in Section \ref{sec:model}, and works for any resource allocation algorithm, since the delay matrices $\bold{T_{\text{avg}}}$ and $\bold{D}$ can be evaluated via estimation or learning methods. 

Finally, we note that key notations used throughout the analysis heretofore are listed in Table \ref{tabel:notation} for ease in reference in the paper.

\begin{table}[htbp]
	\caption{Notations}
	\label{tabel:notation}
	\begin{center}
		\resizebox{\columnwidth}{!}{
		\begin{tabular}{|l|l|}
			\hline
			Notation & Description\\ \hline
			$M$ & Number of files \\ \hline
			$\mathcal{F}_i$ & File name \\ \hline
			$N$ & Number of users \\ \hline
			$U_{i}$ & User name \\ \hline
			$F$ & Size of files \\ \hline
			$\mu$ & Cache size \\ \hline
			$\mathbf{\Phi}$ & Cache state \\ \hline
			$\phi_{i,j}$ & Cache state indicator \\ \hline
			$T_0$ & Duration of each time frame \\ \hline
			$\kappa$ & Index of time frame \\ \hline
			$t_{\kappa}$ & Time instant when $\kappa^{\tth}$ time frame begins \\ \hline
			$T_p$ & Users' behavior period \\ \hline
			$\lambda_{i,j}(t)$ & Intensity function of requests generated by $U_i$ for\\&  $\mathcal{F}_j$ \\ \hline
			$\upsilon$ & Index of operation cycle \\ \hline
			$\tau_{\upsilon}$ & Time instant in the $\upsilon^{\tth}$ time frame \\ \hline
			$\E\left\{\Nreq^{i,j}[\upsilon]\right\}$ & Average number of requests \\ \hline
			$t_{\upsilon}$ & Time instant when $\upsilon^{\tth}$ operation cycle begins \\ \hline
			$N_c$ & Number of channels in OFDMA system \\ \hline
			$B$ & Bandwidth of each channel \\ \hline
			$\E\{T_{i,j}\}$ & Average transmission delay of link $U_i-U_j$ \\ \hline
			$\E\{T_{i,i}\}$ & Average transmission delay between $U_i$ and\\& base station \\ \hline
			$\bold{T_{\text{avg}}}$ & $N\times N$ average transmission delay\\& matrix \\ \hline
			$S$ & Best source table \\ \hline
			$D$ & Smallest average transmission delay to deliver $\mathcal{F}_j$\\& to $U_i$ \\ \hline
			$\eta[\upsilon]$ & average transmission delay of the requests\\& generated in the $\upsilon^{\tth}$ operation cycle\\ \hline
			$\omega_{i,j}^{\upsilon}$ & Weight of request generated by $U_i$ requesting\\& $\mathcal{F}_j$ \\ \hline
			$N_T$ & Number of periods \\ \hline
			$N_p$ & Number of samples collected over $N_T$ periods \\ \hline
			
		\end{tabular}
	}
	\end{center}
\end{table}

\subsection{Complexity analysis of the proposed caching algorithm}\label{complexity}

	The proposed caching algorithm is able to reduce the search space and obtain a caching policy with delay performance close to the optimal solution. According to the proposed caching algorithm, for each operation cycle, the base station runs Algorithm 2 for $N\mu$ times to find $N\mu$ $<User, File>$ pairs to update all users' cache space. At the $l^{\text{th}}$ time the base station runs Algorithm 2, it searches over $NM-(l-1)$ possible $<$file,user$>$ pairs, where the term $l-1$ corresponds to the $l-1$ $<$file,user$>$ pairs that have been selected in previous iterations. Therefore, the size of the search space of our algorithm is $\sum_{l=1}^{N\mu} NM-(l-1)=N^2M\mu-\frac{1}{2}N^2\mu^2+\frac{1}{2}N\mu$, which is much smaller than the size of the entire solution space $(\frac{M!}{(M-\mu)!\,\mu!})^N$.

\section{Broadcasting in the Transmission Phase} \label{Sec:broadcast}
In the analysis of the transmission delay in Section \ref{Sec:delay}, we have considered point-to-point links and assumed that each  transmitter can only transmit to one receiver at a time. While the proposed caching algorithm is applicable for any scheme as long as average delays can be estimated, our numerical results in Section \ref{Sec:numerical} mainly considers such point-to-point transmissions.
In this section, we extend our analysis and study how the transmission delay is influenced when broadcasting is allowed. Analysis is conducted with the following assumptions:
\begin{itemize}
	\item At each time instant, each user will generate one request for a file based on its own preference for files.
	\item If a file is only requested by one user, the algorithm will first check if the file is cached at any other user. If yes, the link with smallest delay will be selected to transmit; otherwise, the base station will transmit the file to the corresponding user. And the delay caused by this file is the average transmission delay between the source and the receiver.
	\item When a file is requested by more than one user, then the base station or a user with the cached file will broadcast to all other users who do not have the file but have requested it. In the selection of the source for broadcasting, maximization of the minimum of the rates in the links used for broadcasting will be considered as detailed below.
	\item In a time slot, D2D transmitters are not allowed to broadcast more than one file simultaneously.
\end{itemize}

Now, we introduce how the transmission delay is determined when a file is broadcast to multiple receivers. In every time slot, after the caching decisions are made, the transmission phase starts when the actual requests for files are generated by the users. Assume that in the $\kappa^{\text{th}}$ time frame, the channel capacity for a transmission link between the transmitter $U_{t}$ (which can include the base station) and the receiver $U_{r}$  is
\begin{align}
C_{t, r}[\kappa]= B\log_2\left(1+\frac{P_t}{B\sigma^2}z_{t,r}[\kappa] \right).
\end{align}

Then assume that for a file $\mathcal{F}$, there are $N_t$ users that can broadcast it to $N_r$ users. If transmitter $U_T$ where $T = 1, 2, ..., N_t$ is selected to broadcast the file, the transmission rate should be the minimum of the rates in the links between $U_T$ and all receivers so that all receivers can reliably get the file. Hence, we have
\begin{align}
{C_{T}}[\kappa ] =&& \min \left\{ {B{{\log }_2}\left( {1 + \frac{{{P_t}}}{{B{\sigma ^2}}}{z_{T,1}}[\kappa ]} \right),} \right.  \nonumber\\
&& \left. {\ldots,B{{\log }_2}\left( {1 + \frac{{{P_t}}}{{B{\sigma ^2}}}{z_{T,{N_r}}}[\kappa ]} \right)} \right\}.
\end{align}

Now, in order to achieve the smallest delay, we select, among the $N_t$ source candidates, the user with the largest $C_T[\kappa]$. Hence, the broadcast rate is
\begin{align}
{\boldmath{C_j}}[\kappa]= \max  \{  {{C_{1}}[\kappa ], {C_{2}}[\kappa ], \ldots, {C_{N_t}}[\kappa ]} \}. \label{eq:broadcastrate}
\end{align}
Note that we essentially need to solve a max-min problem in determining the source that will broadcast the file.

Having characterized the broadcast rates, we can now follow the same approach described in Section \ref{Sec:delay} to determine the average delays in each link. For instance, the duration to send a file can be determined via the formulation in (\ref{eq:transmissiondelay}) using the rate in (\ref{eq:broadcastrate}), and average of these durations can be determined through numerical and simulation results.  An important distinction is that in the computation of the overall average system delay, only a single delay term needs to be used when a file is broadcast to multiple users (instead of summing up the delays to transmissions to different users because the file is sent simultaneously to multiple users in the broadcast approach).

%

\section{Numerical Results}\label{Sec:numerical}


In this section, we investigate the performance of the proposed intensity estimation algorithm and the caching algorithm via numerical results. In Section \ref{EstResult}, we discuss the performance of intensity estimation in terms of the estimation error. In Section \ref{DelayResult}, we compare the performance of the proposed caching algorithm with those of the following algorithms:

\begin{itemize}
	\item Naive algorithm: In this algorithm, each user just caches the most popular $\mu$ files. In our implementation, the naive algorithm learns the popularity of files from the users' arrival intensity function for each file. Naive algorithm can be an effective approach when the base station does not have the knowledge of the channel fading statistics and the cached files at each user. With this algorithm, the users just cache files according to their own preferences. The comparisons are provided in three parts. In Section \ref{ct}, we show the performance differences between our proposed algorithm and naive algorithm. Then, the average delay $\eta$ is considered for both caching algorithms in two scenarios when the users' preferences are unknown and known to the system in Section \ref{prfct} and Section \ref{pplrty}, respectively.
	
	\item Probabilistic algorithm: Probabilistic caching algorithms have recently been addressed for instance in  \cite{chen2017optimal} and \cite{chen2017probabilistic}.
Different from the idea that each user caches the most popular files based on its own preference, the probabilistic algorithm introduces randomness into caching decisions, which potentially increases the chances that the users collaborate with each other to reduce the overall transmission delay via D2D transmissions. Specifically, in probabilistic caching, a file is cached with a probability that is proportional to its popularity. Hence, highly popular files have a higher chance to be cached while less popular files can still be cached but with relatively low probabilities. We compare the proposed algorithm with the probabilistic algorithm in the Section \ref{pplrty}, where the popularity matrix of the files is known to the system and is described by the Zipf distribution.

	
\end{itemize}

\subsection{Performance of the Intensity Estimation Algorithm}\label{EstResult}
Intuitively, the performance of intensity estimation has positive correlation with the number of observation samples, i.e., the more samples we have, the smaller error of intensity estimation will be. In our case, for the intensity function given in (\ref{eq:int_req1}), the number of samples is decided by the intensity function and the number of periods $N_T$. However, the number of requests we collect is dependent on the users' demand, which can be very large or small. And we cannot afford to arbitrarily increase the number of periods over which we collect samples because of potential high cost. So, in this part we study how the error in intensity estimation varies with the minimum intensity, which is defined as the average number of requests for the least popular file, and also as $N_T$ changes.

In Figure \ref{fig:est_1}, we set the number of users as $N=25$, the number of files in the library as $M = 100$, and the cache size as $\mu = 30$. And we perform the experiments with the number of period to sampling set as $N_T=\{1, 10, 100\}$ respectively, and plot the intensity estimation error as a function of the minimum intensity. As the minimum intensity increases, the error in intensity estimation decreases because the minimum intensity controls the lower bound on the requests for each file. For different values of $N_T$, we note that all error curves tend to converge at a point when the minimum intensity is sufficiently large, i.e., when the collected samples are sufficient for estimation. And the error converges faster when we have a larger $N_T$.


\begin{figure}
\centering
\includegraphics[width=\figsize\linewidth]{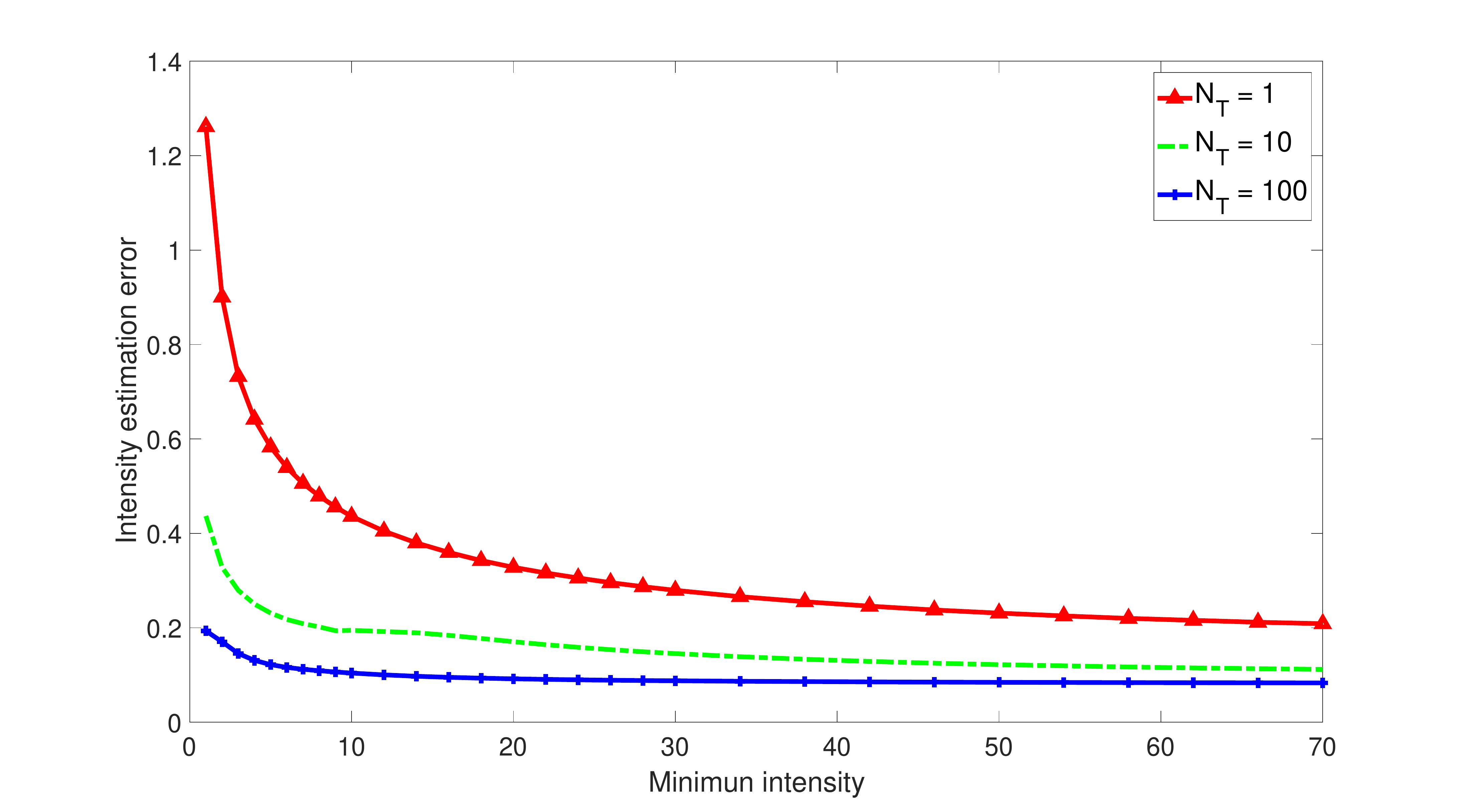}
\caption{Intensity estimation error vs. minimum intensity for $N_T=1,10,100$. The number of files in the library is $M = 100$, and the cache size is $\mu = 30$.}\label{fig:est_1}
\label{fig:est_1}
\end{figure}

In Figure \ref{fig:est_2}, we set the number of users as $N=7$, the number of files in the library as $M = 15$, and the cache size as $\mu = 4$, and we perform the experiments with the value of minimum intensity set as $\{40, 100, 400\}$ respectively,  and plot the intensity estimation error as a function of $N_T$. As $N_T$ increases, the error in intensity estimation decreases. When $N_T$ exceeds a certain threshold, the error in intensity estimation converges and does not diminish further because $N_T$ is already large enough to provide a sufficient number of samples to the system. For different values of the minimum intensity, when the three curves are compared, we observe that the threshold for convergence is smaller when we have a larger value for the minimum intensity.

\begin{figure}
	\centering
	 \includegraphics[width=\figsize\linewidth]{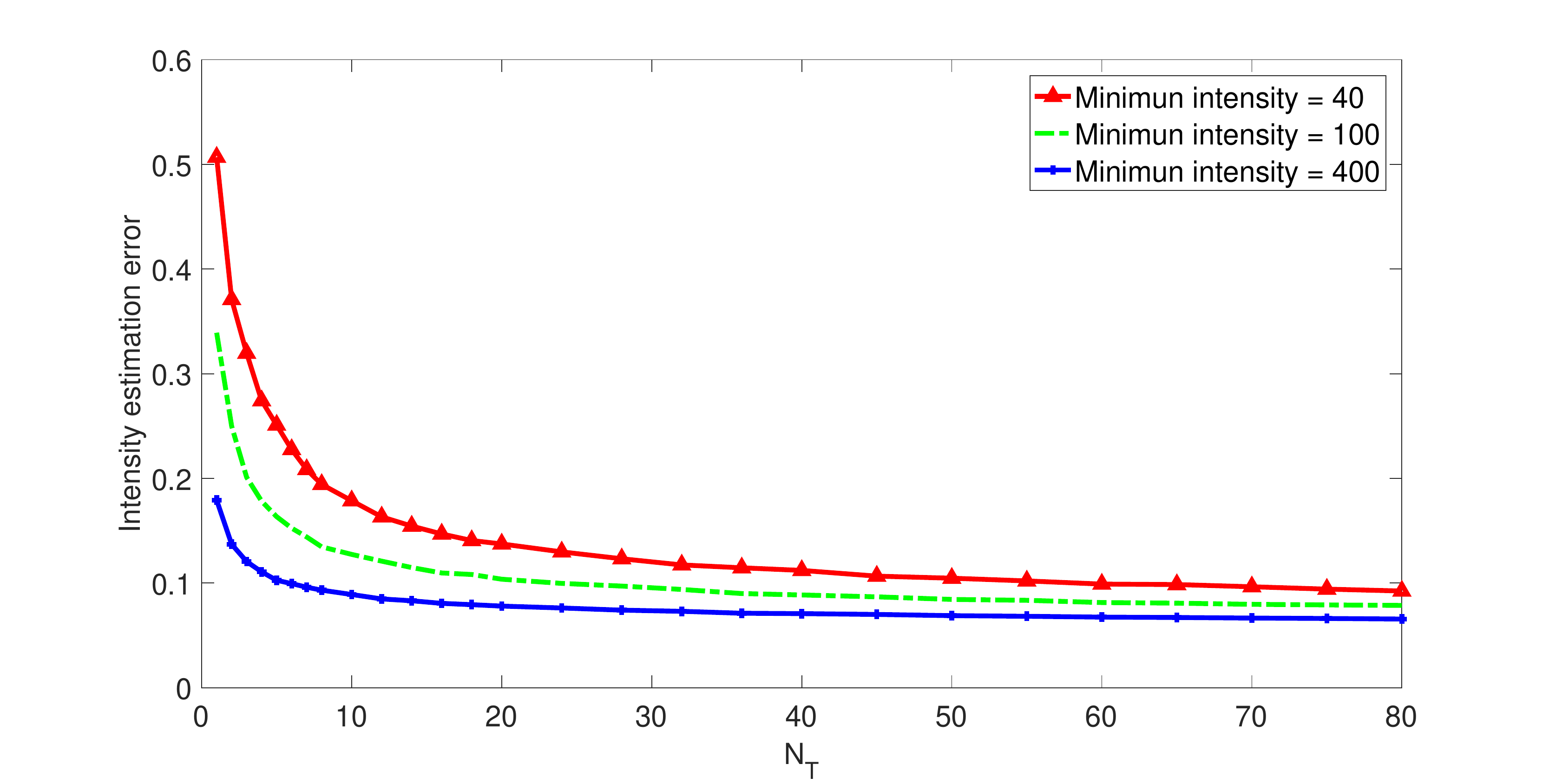}
	\caption{Intensity estimation error vs. $N_{T}$ for minimum intensity = $40,100,400$. The number of files in the library is $M = 15$, and the cache size is $\mu = 4$}\label{fig:est_2}
	\label{fig:tnumerror}
\end{figure}


\subsection{Performance of the Caching Algorithm}\label{DelayResult}

Now, we first show the difference between the proposed caching algorithm and naive algorithm in Section \ref{ct}, and then we compare the performance of the proposed algorithm with other algorithms in terms of average delay in Sections \ref{prfct}, \ref{bc} and \ref{pplrty}.

As shown in Figure \ref{fig:system}, the users are located within a circular cell with the base station placed at the center. In Section \ref{ct}, we design a scenario where the distances between users are less than the distances between the users and base station, so that we can force the users to select D2D transmission rather than the cellular transmission. Also, the users' preference for files are all designed for the purpose to track the cache state of each user and see how the proposed algorithm works differently when compared to the naive algorithm.  In Sections \ref{prfct}, \ref{bc} and \ref{pplrty}, the location and the preference for files of each user are randomly generated. In our experiments, we assume Rayleigh fading with path loss $\E\{z\}=d^{-4}$, where $d$ represents the distance between the transmitter and the receiver, and $\eta$ denotes the average system delay.


\subsubsection{Performance comparison between the proposed caching algorithm and naive algorithm}\label{ct}

In our implementation in this part, Section \ref{prfct} and \ref{bc}, we set the cell radius as $R = 1.5 km$, base station transmission power as $P_b = 16.9 dB$, users' transmission power as $P_u = 13 dB$, and package size of each file as $F = 96.13 bits$.

To observe the difference between the proposed caching algorithm and naive algorithm, we design a experiment with $N = 3$, $M = 21$, and $\mu = 3$. Then, we set the preferences of the three users (or more explicitly the most popular nine files for each user $U_1$, $U_2$, and $U_3$ in decreasing popularity) as $\{\mathcal{F}_1,\mathcal{F}_2,\mathcal{F}_3,\mathcal{F}_4,\mathcal{F}_5,\mathcal{F}_6,\mathcal{F}_7,\mathcal{F}_8,\mathcal{F}_9\}$, $\{\mathcal{F}_8,\mathcal{F}_9,\mathcal{F}_{10},\mathcal{F}_{11},\mathcal{F}_{12},\mathcal{F}_{13},\mathcal{F}_{14},\mathcal{F}_1,\mathcal{F}_2\}$ and
\\
$\{\mathcal{F}_{15},\mathcal{F}_{16},\mathcal{F}_{17},\mathcal{F}_{18},\mathcal{F}_{19},\mathcal{F}_{20},\mathcal{F}_{21},\mathcal{F}_1,\mathcal{F}_2\}$, respectively. To make sure that the users have the chance to work collaboratively, we have the users to become active alternately as follows: $U_1$ and $U_3$ are most active  \footnote{The most active means that the number of the requests by the user will reach its peak in the current period of the intensity function.} in the $25^{\tth}$ operation cycle and $U_2$ is most active in $75^{\tth}$ operation cycle.

Table \ref{tabe:cache state} shows the cache states of the three users in the $1^{\tth}$, $25^{\tth}$, and $75^{\tth} $ operation cycles when the proposed algorithm is employed. We notice that in the $1^{\tth}$ operation cycle, all users cache their own favorite files. In the $25^{\tth}$ operation cycle, under the naive algorithm, all users will still cache files based their own preferences. But under the proposed algorithm, $U_1$ and $U_3$ cache their own favorite files, and since $U_2$ is not active in this operation cycle, it caches $\mathcal{F}_{4}$ and $\mathcal{F}_{5}$ for $U_1$, and $\mathcal{F}_{18}$ for $U_3$. In the $75^{\tth}$ operation cycle, under the naive algorithm, users will keep caching their own favorite files while under the proposed algorithm, $U_2$ caches its own favorite files, and since $U_1$ and $U_3$ are not active at this time, they both help to cache $U_2$'s remaining favorite files. This example demonstrates that with the proposed caching algorithm, users collaborate to reduce the overall average delay, but with the naive algorithm, each user attempts to only reduce its own average delay.\\

%


\begin{table}
	\centering
	\caption{Cache states}
	\label{tabe:cache state}
\small
\begin{tabular}{|c|c|c|c|}
\hline
	index of operation cycle & user $1$ & user 2& user 3\\ \hline
	1 & $\{1,2,3\}$ & $\{8,9,10\}$ & $\{15,16,17\}$ \\ \hline
	25 & $\{1,2,3\}$ & $\{4,18,5\}$ & $\{15,16,17\}$ \\ \hline
	75 & $\{11,12,13\}$ & $\{8,9,10\}$ & $\{14,1,2\}$ \\
\hline
\end{tabular}
\end{table}

\subsubsection{Average delay based on intensity estimation}\label{prfct}

Now, we compare the proposed caching algorithm with naive algorithm based on the real intensity and estimated intensity obtained by the proposed non-parametric estimator. In the experiment, we describe the real intensity as the perfect intensity function, and the estimated intensity as the imperfect intensity function. In this part, the total number of files is set as $M = 100$.

In Figure \ref{fig:delay_mu}, we set $N=25$ and plot the overall average delay as a function of the cache size $\mu$. When $\mu$ is small, the proposed algorithm does not lead to a significant advantage over the naive algorithm. As $\mu$ increases, the average delay $\eta$ of both algorithms decrease, and the gap between the proposed algorithm and naive algorithm grows because the users collaborate in the proposed algorithm while in the naive algorithm, users only consider reducing their own average delay. As $\mu$ increases further and exceeds a threshold, the gap between the two algorithms decreases, because the system starts having sufficient cache space for all popular files. And the curves with perfect intensity function and the imperfect intensity function for both algorithms are almost overlapping, which demonstrates the superior performance of the non-parametric estimator.

\begin{figure}
	\centering
	 \includegraphics[width=\figsize\linewidth]{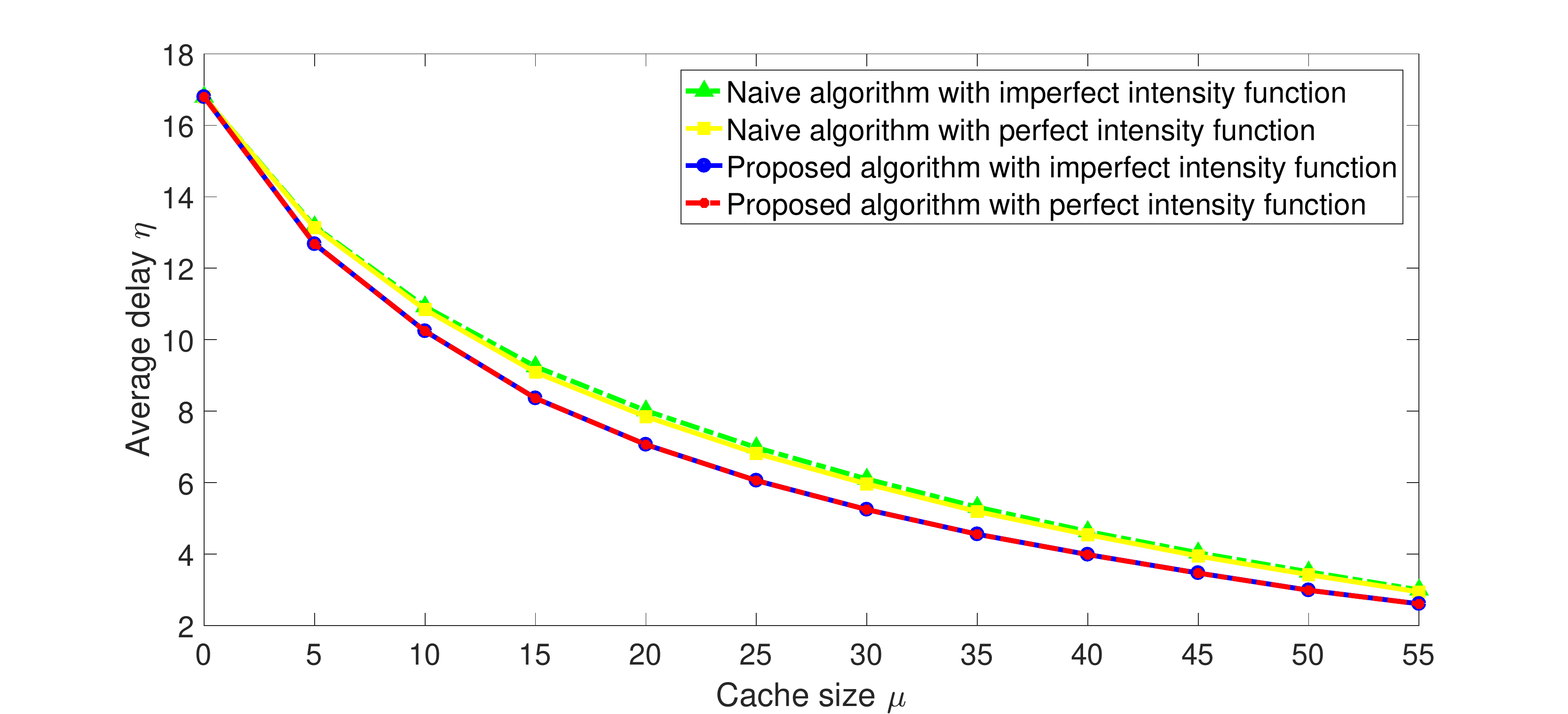}
	\caption{Average delay $\eta$ vs. cache size $\mu$}
	\label{fig:delay_mu}
\end{figure}

In Figure \ref{fig:delay_N}, we set $\mu=25$ and plot the overall average delay as a function of the number of users $N$. When $N = 1$, the only user cannot select the D2D mode, so the average delay $\eta$ is the transmission delay from base station to the user. As $N$ increases, the average delay $\eta$ tends to decrease, and the gap between the proposed algorithm and naive algorithm increases rapidly because having more users means there will be more collaboration based on other users' preferences in the proposed algorithm which can provide the optimal $<$file,user$>$ pairings and lead to the maximum delay improvement. And comparing the performance gap between the cases with the perfect intensity function and imperfect intensity function, we find that the proposed caching algorithm has better tolerance to estimation errors.


\begin{figure}
	\centering
	 \includegraphics[width=\figsize\linewidth]{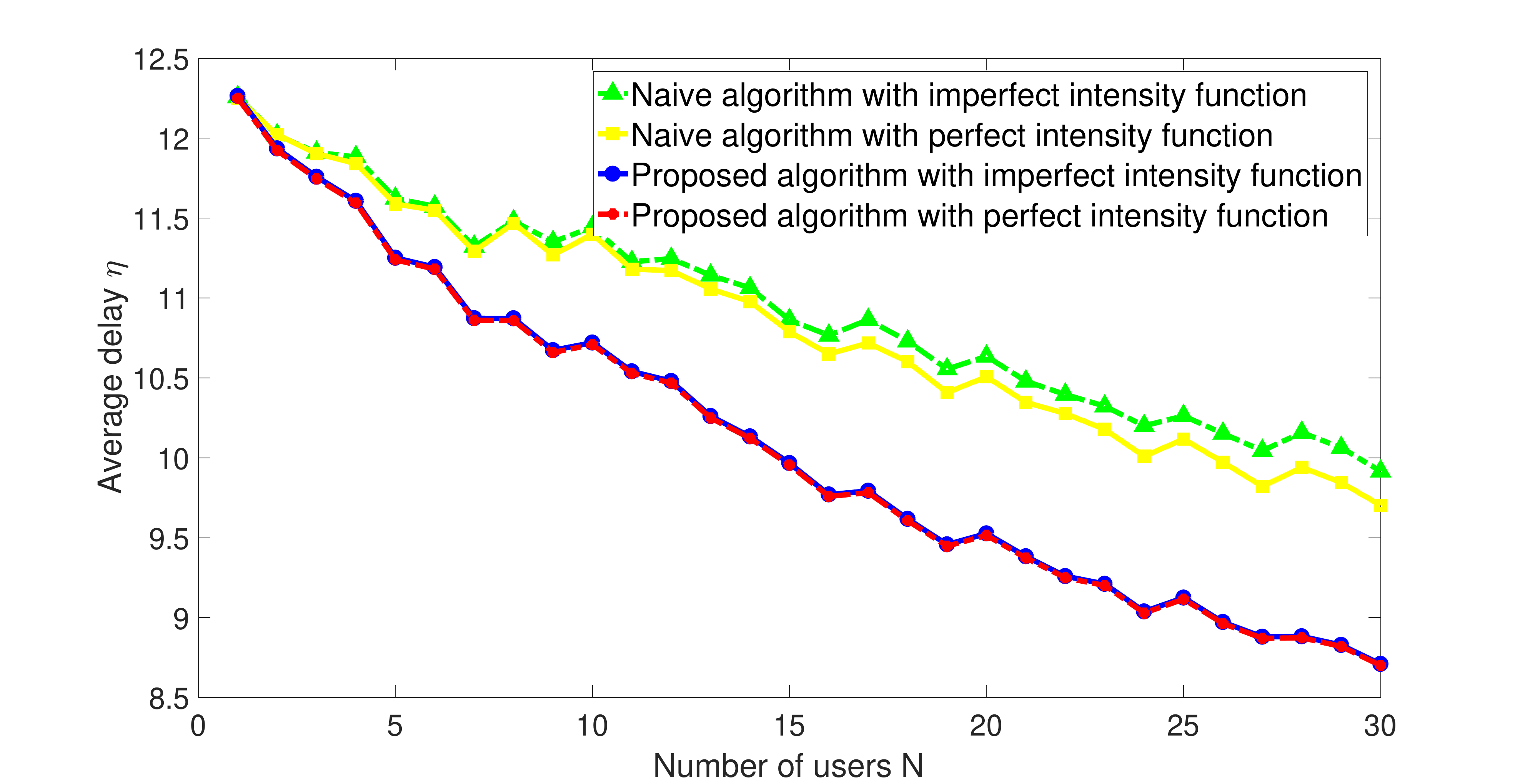}
	\caption{Average delay $\eta$ vs. the number of users $N$}
	\label{fig:delay_N}
\end{figure}

\subsubsection{Average delay with broadcasting in transmission phase}\label{bc}

In this part, we compare the average delay in scenarios in which the broadcasting is allowed and not allowed in the transmission phase. In the simulations, we set select $N = 25$, $M = 100$, $\mu = 30$, and $\beta = 0.25$. First, we run the estimation and caching algorithm to update the caching state at the beginning of each operation cycle. After all users' caches are updated, we start the simulation of the transmission phase. For each time frame $t_{\kappa}$, each user generates $1000$ file requests based on the arrival function estimated in the current operation cycle.

Table \ref{tabel:broadcast} provides the average system delay in simulations for both the proposed algorithm and the naive algorithm. For both caching algorithms, broadcasting can help to reduce the average delay. And the simulation results also verify that our proposed algorithm can work under different transmission schemes. We also notice that reduction in transmission delays are limited in both algorithms, because broadcasting requires the circumstances in which requests for the same content arrive simultaneously, which heavily depends on the user density and content popularity distribution.


\begin{table}[]
	\caption{Average delay $\hat{\eta}$}
	\label{tabel:broadcast}
	\centering
	\resizebox{\columnwidth}{!}{
		\begin{tabular}{cclclc}
			\cline{1-4}
			\multicolumn{1}{|c|}{}                   & \multicolumn{2}{c|}{Broadcasting allowed} & \multicolumn{1}{c|}{Only one link allowed} &  &                      \\ \cline{1-4}
			\multicolumn{1}{|c|}{Proposed Algorithm} & \multicolumn{2}{c|}{6.9358}               & \multicolumn{1}{c|}{7.1004}                &  &                      \\ \cline{1-4}
			\multicolumn{1}{|c|}{Naive Algorithm}    & \multicolumn{2}{c|}{7.4618}               & \multicolumn{1}{c|}{7.5320}                &  &                      \\ \cline{1-4}
			\multicolumn{1}{l}{}                     & \multicolumn{1}{l}{}          &           & \multicolumn{1}{l}{}                       &  & \multicolumn{1}{l}{}
		\end{tabular}
	}
\end{table}

\subsubsection{Average delay based on popularity models}\label{pplrty}

In our implementation in this part, the total number of files is $M = 100$, and we set the cell radius as $R = 1.8$ km, base station transmission power as $P_b = 23$ dB, users' transmission power as $P_u = 20$ dB, and package size of each file as $F = 11.29$ bits.

In this part, we extend the comparison of the proposed caching algorithm, naive algorithm and probabilistic algorithm to different popularity models. We assume in this case that the users' preferences are known to the system, and both algorithms can have direct access to the popularity rank of files. Hence, intensity estimation is not applied. The requests are generated based on Zipf distribution.

Here, we discuss the performance in terms of overall average delay. We consider two cases: independent popularity, in which users have independent preferences, and identical popularity model, in which users have identical preferences for files. In the case of independent preferences, each user has a different rank for the files, so the cache at each user could be different.  

In the case of naive algorithm with identical preferences, users get the files they do not have via cellular downlink from the base station. Therefore, the gap between the two curves using naive algorithm in Figs. \ref{fig:delay_beta} - \ref{fig:delay_N_old} (which will be discussed in detail next) demonstrates the benefit of enabling D2D communications. By allowing D2D transmission, the users far away from the base station can get files from their neighbors, which helps to significantly reduce the delay. Also, as shown in Figs. \ref{fig:delay_beta} - \ref{fig:delay_N_old}, the performance of the probabilistic algorithm for both popularity models is always in between the performances of the naive algorithm achieved with the identical popularity model and independent popularity model.  As discussed above, in the case of identical popularities, naive algorithm actually does not lead to file exchanges via D2D links because all users store the same set of files (i.e., the most popular $\mu$ files, which are the same for all users). However, with the probabilistic algorithm, though the caching decisions highly rely on the popularity distribution of the files, there are chances for less popular files to be cached. Due to this, different files can be cached by different users, which increases the possibility to reduce the overall transmission delay via D2D transmissions. In the case of the independent popularity model, the naive algorithm assigns different top ranked files to users' cache space based on their preferences, and the transmission delay for each user's most popular files is guaranteed to be reduced and also D2D links can be utilized for file exchange. On the other hand, the probabilistic algorithm risks missing the most popular files due to its probabilistic nature, and so the performance is comparatively worse. Finally, as will be discussed next, the proposed delay-aware caching algorithm outperforms both naive and probabilistic caching algorithms.



In Fig. \ref{fig:delay_beta}, we set $N=25$, $\mu=30$ and plot the average delay $\eta$ as a function of the Zipf exponent $\beta$. As $\beta$ increases, the popularity difference increases. When $\beta=0$, the users request all files with equal probability; when $\beta\rightarrow +\infty$, each user only requests its most favorite file. Therefore, we only need to concentrate on the delay performance of fewer popular files as $\beta$ increases, and it becomes easier to achieve better delay performance with limited caching space. That is the reason for having monotonically decreasing curves in Fig. \ref{fig:delay_beta}. Another observation is that our algorithm is very robust to the popularity setting. Compared to the curves using the naive algorithm, identical popularity model only slightly raises the delay of our algorithm. If a node can get a popular file from its near neighbor, then caching some less popular files might give better delay improvement. Therefore, our algorithm can enable D2D transmission even in an identical popularity model, which guarantees the robustness and low transmission delay. As to the probabilistic algorithm, though the identical popularity model does not worsen its performance, the average delay achieved by the probabilistic algorithm is not competitive compared to the proposed algorithm.

\begin{figure}
\centering
\includegraphics[width=\figsize\linewidth]{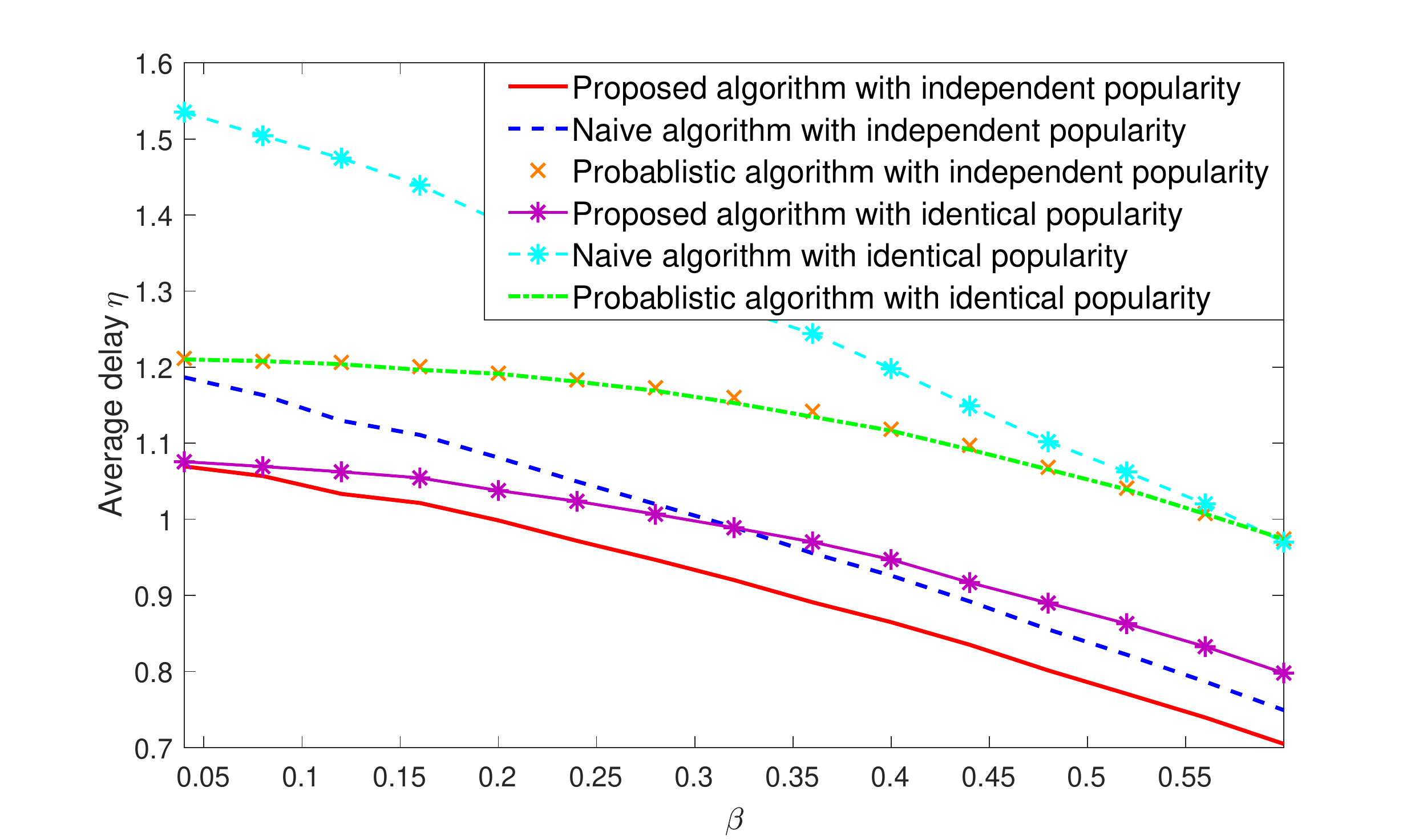}
\caption{Average delay $\eta$ vs. Zipf exponent $\beta$}
\label{fig:delay_beta}
\end{figure}


In Fig. \ref{fig:Mu_study_latest}, we select $\beta=0.1$, $N=25$ and plot the average delay as a function of the cache size $\mu$. When $\mu$ is small, the delay difference between different algorithms and different popularity settings is small. In such a situation, both algorithms cache the most popular files. As $\mu$ increases, the difference in performance increases. As we have mentioned in Algorithm \ref{Algorithm2}, our algorithm searches for the optimal $<$file,user$>$ pair that provides the maximum delay improvement, and this mechanism guarantees a very sharp decrease at the beginning. After exceeding a threshold, further increasing the caching size reduces the performance difference, because the system gets enough caching size to cache most of the popular files. Overall, Fig. \ref{fig:Mu_study_latest} shows that our algorithm can achieve better delay performance with limited cache size.


\begin{figure}
\centering
\includegraphics[width=\figsize\linewidth]{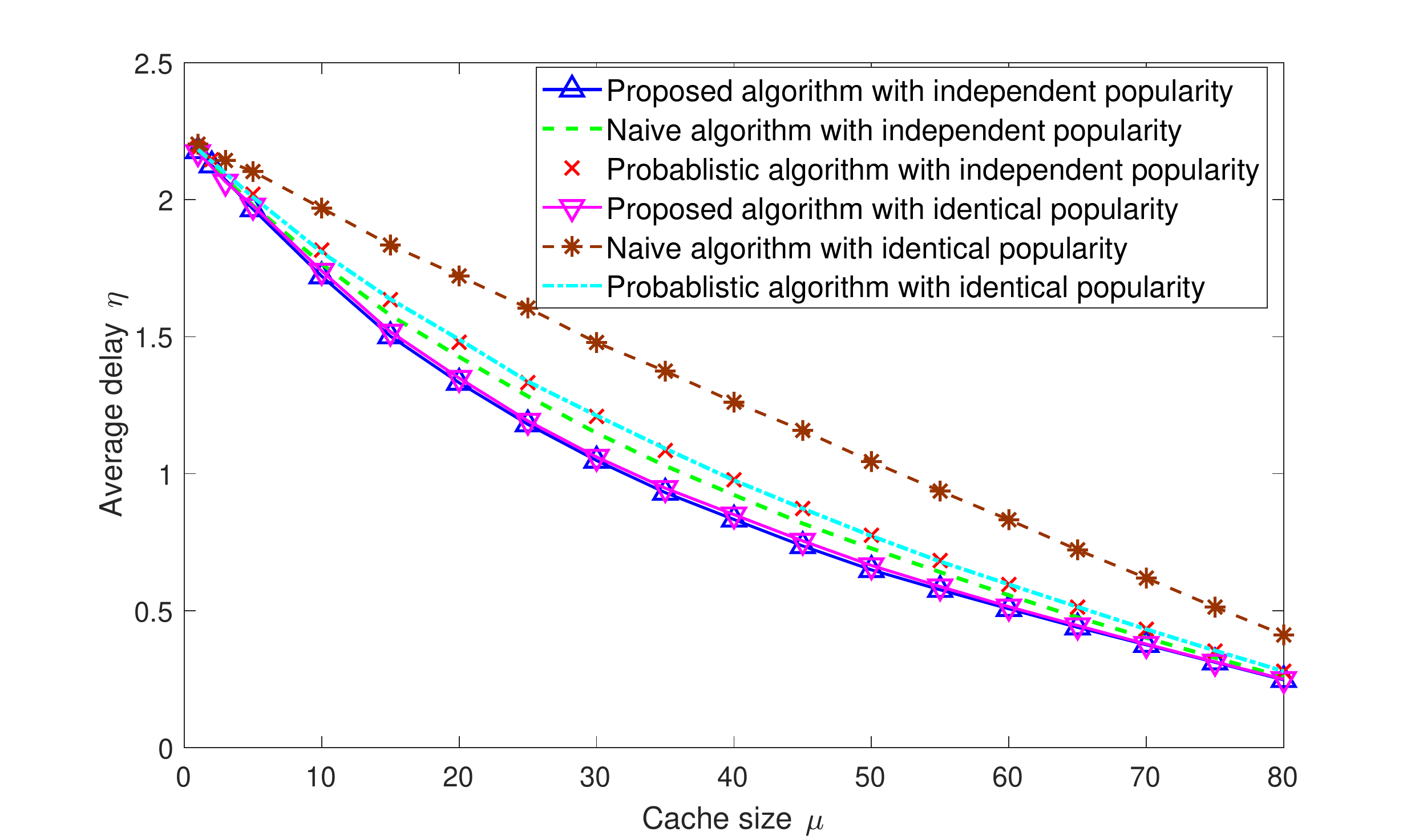}
\caption{Average delay $\eta$ vs. cache size $\mu$}
\label{fig:Mu_study_latest}
\end{figure}


In Fig. \ref{fig:delay_N_old}, we select $\beta=0.1$, $\mu=30$ and plot the average delay as a function of the number of users $N$. For the curve using the naive algorithm with identical popularity model, having more users does not affect the average delay because each user works in cellular mode and receives the files from the base station. For other curves, increased number of users enables more chances for D2D communication, and as a result the average delay decreases. Compared with the naive algorithm and probabilistic algorithm, our algorithm can achieve better performance, especially when the number of users is large.



\begin{figure}
\centering
\includegraphics[width=\figsize\linewidth]{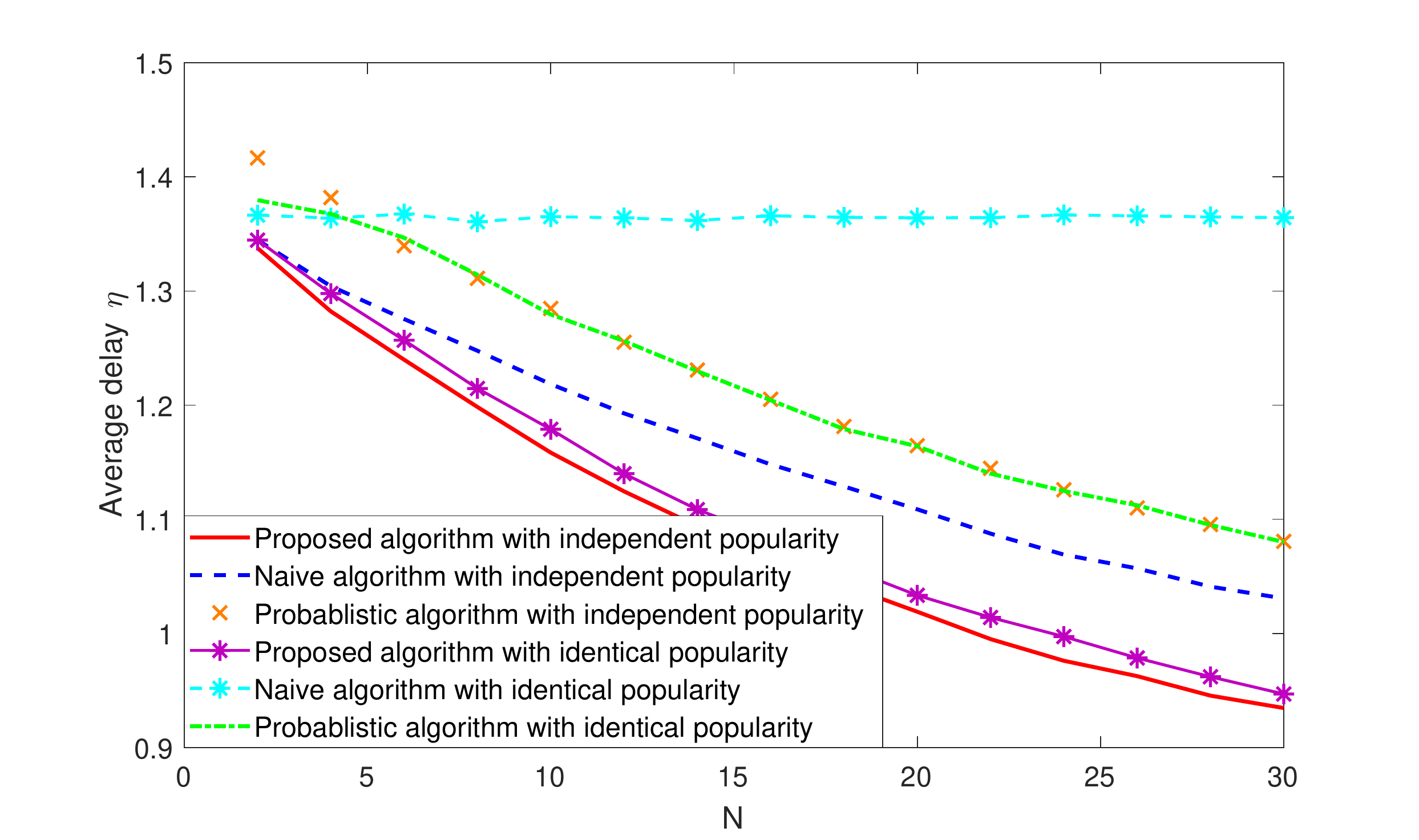}
\caption{Average delay $\eta$ vs. the number of users $N$}
\label{fig:delay_N_old}
\end{figure}

\section{Conclusion} \label{Sec:conclusion}
In this paper, we have proposed a learning-based caching algorithm for D2D cellular networks, which minimizes the weighted average delay. First, we have learned the intensity function of the users' requests by using a kernel estimator and characterized average transmission delay of a request. Then, we have formulated the delay minimization problem and developed our algorithm which can solve the weighted average delay minimization problem efficiently in a general scenario in which the distributions of fading coefficients and system parameters change over time. We  have investigated the performance of both the kernel estimator and caching algorithm. To demonstrate the performance of the kernel estimator, we analyzed the estimation error with changing minimum intensity and the number of time frames. To investigate the performance of the proposed caching algorithm, we have compared it with a naive algorithm which simply caches the most popular files at each user, and a probabilistic algorithm in which the users cache files based on their popularities. And the comparisons have been made on models with unknown and known popularities of the content files. For the model with unknown popularity information, our caching algorithm operates with the kernel estimator. For the model with known popularity, the caching algorithms have been applied to two different popularity models. For both scenarios with and without the popularity information, we have shown that the proposed algorithm is more robust to variations in the popularity models, and can achieve better performance because the proposed algorithm can more effectively take advantage of D2D communications. Also, the impact of the popularity parameter, caching size and number of users is further identified via numerical results.

\bibliographystyle{ieeetr}
\bibliography{D2D_caching}

\end{document}